\documentclass{article}%
\usepackage{amsmath}
\usepackage{amstext}
\usepackage{amsfonts}
\usepackage{amssymb}
\usepackage{amsopn}
\usepackage{graphicx}
\usepackage{cite}
\usepackage{bm,bbm}
\usepackage{epsfig,epic}%
\newtheorem{theorem}{Theorem}
\newtheorem{lemma}[theorem]{Lemma}
\newtheorem{corollary}[theorem]{Corollary}
\newtheorem{proposition}[theorem]{Proposition}
\newenvironment{proof}[1][Proof]{\noindent\textbf{#1.} }{\ \rule{0.5em}{0.5em}}
\begin{document}


\vspace{1cm}

\begin{center}

{\Large Volume of the set of unistochastic matrices of order $3$ \newline and
the mean Jarlskog invariant}
\end{center}

\vspace{1cm}

{\large Charles Dunkl} \footnote{Email address: cfd5z@virginia.edu}
\newline
Department of Mathematics, University of Virginia, \newline Charlottesville,
VA 22904-4137, USA 
\\

\vspace{0.1cm}

{\large Karol \.Zyczkowski} \footnote{Email address: karol@tatry.if.uj.edu.pl}
\newline Institute of Physics, Jagiellonian University, Cracow \newline and
\newline Center for Theoretical Physics, Polish Academy of Sciences, Warsaw

\vspace{12mm}

November 10, 2009

\vspace{13mm}

\textbf{Abstract}


\vspace{5mm}

\noindent A bistochastic matrix $B$ of size $N$ is called
\textsl{unistochastic}  if there exists a unitary $U$ such that $B_{ij}%
=|U_{ij}|^{2}$  for $i,j=1,\dots,N$.  The set $\mathcal{U}_{3}$ of all
unistochastic matrices  of order $N=3$ forms a proper subset of the Birkhoff
polytope,  which contains all bistochastic (doubly stochastic) matrices.  We
compute the volume of the set $\mathcal{U}_{3}$ with respect to  the flat
(Lebesgue) measure and analytically evaluate  the mean entropy of an
unistochastic matrix of this order.  We also analyze the Jarlskog invariant
$J$, defined for any  unitary matrix of order three,
and derive its probability distribution for the ensemble of  matrices
distributed with respect to the Haar measure on $U(3)$  and for the ensemble
which generates the flat measure  on the set of unistochastic matrices.  For
both measures the probability  of finding $|J|$ smaller than the value
observed for the CKM matrix,  which describes the violation of the CP parity,
is shown to be small.  Similar statistical reasoning may also be applied to
the MNS matrix, which plays role in describing the neutrino oscillations. Some
conjectures are made concerning analogous probability measures in the space of
unitary matrices in higher dimensions.

\newpage

\section{Introduction}

\vspace{5mm}

Bistochastic matrices appear in variety of problems in different branches of
science. A bistochastic matrix (also called doubly stochastic) contains real
non-negative entries, the sum of which in each column and in each row is equal
to unity. Thus each column and each row of such a matrix can be interpreted as
a probability vector. The structure of the set $\mathcal{B}_{N}$ of all
bistochastic matrices of order $N$ is well understood \cite{MO79}. It is
formed by the convex polytope of all permutation matrices of size $N$ and it
is often called the \textsl{Birkhoff polytope} \cite{Bi46}. Analytical
expressions for its volume for small dimensionality \cite{CR99,BP03}
and for the leading terms of the volume in the asymptotic limit
\cite{CM07,CBSZ08} are available in the literature.

In various physical problems it is assumed that the probabilities, which form
the entries of a bistochastic matrix, arise as squared modulus of an element
of an (a priori unknown) unitary matrix. A bistochastic matrix $B$ which can
be generated from an unitary matrix $U$ by the relation $B_{ij}=|U_{ij}|^{2}$
is called \textsl{unistochastic} (or orthostochastic).

For instance, such matrices are used in high energy physics to characterize
interactions between elementary particles, which can be divided into $N$
generations. Since the Hamiltonians describing two kinds of physical
interactions (usually called 'strong' and 'weak') do not commute, these two
Hermitian operators determine a single unitary matrix of order $N$, which
relates both eigenbases. This is the famous unitary matrix of
Cabbibo-Kobayashi-Maskawa \cite{KM73}. Note that the squared moduli of the CKM
matrix $V_{\mathrm{CKM}}$ form the corresponding unistochastic matrix $B$, the
entries of which represent probabilities which are accessible experimentally
\cite{Ja85,JS88,Di05,Di08,GGPT08}.

According to the Standard Model of elementary particles there exist three
generations of quarks, thus the case of a direct physical importance are
unistochastic matrices of order $N=3$. On the other hand, the case $N=4$ could
also become relevant in case a fourth generation of quarks should be
discovered \cite{BD87,La89,AMM91}. A similar problem in the neutrino physics
is characterized by the Maki-Nakagawa-Sakata matrix (MNS matrix) \cite{MNS62},
some parameters of which are still quite uncertain. A relation between the MNS
and CKM matrices is studied in \cite{CM00}.

In practice, given a bistochastic matrix $B\in\mathcal{B}_{N}$ it is important
to know, whether it belongs to the set $\mathcal{U}_{N}$ of unistochastic
matrices. If this is the case one would like to describe the set of all
unitary matrices such that $B_{ij}=|U_{ij}|^{2}$ for $i,j=1,\dots, N$. These
questions are of a particular interest for research in foundations of quantum
mechanics and investigation of properties of transition
probabilities\cite{Lande,Rovelli,Lo97}, scattering theory \cite{Mennessier},
quantum counterparts of Markov processes and dynamics on graphs
\cite{Tanner,PTZ03}, and the theory of quantum information processing
\cite{We01}.

Any bistochastic matrix of order two is unistochastic, so both sets coincide,
$\mathcal{U}_{2}=\mathcal{B}_{2}$. The situation differs already for $N=3$. To
show this fact Schur considered the symmetric combination of cycle-three
permutation matrices, $P$ and $P^{-1}=P^{2}$. This matrix
\begin{equation}
B_{S} = \frac{1}{2} (P+P^{2}) = \frac{1}{2} \left[
\begin{array}
[c]{ccc}%
0 & 1 & 1\\
1 & 0 & 1\\
1 & 1 & 0
\end{array}
\right] \label{B3}%
\end{equation}
is clearly bistochastic but it is easy to see that there is no corresponding
unitary matrix. Hence $\mathcal{U}_{3} \subsetneq\mathcal{B}_{3}$ and a
similar relation holds for an arbitrary $N\ge3$. Vaguely speaking, the moduli
of a unitary matrix need to fulfill certain constraints, more stringent than
the obvious fact that the sum of squared moduli in each row (each column) is
equal to unity. As recently analyzed by P. Di{\c t}{\v a} \cite{Di06}, this
simple observation has important consequences for reconstruction of unitary
matrices from experimental data.


A general question, if a given bistochastic matrix $B$ is unistochastic,
remains open, and only partial results are available \cite{ZKSS03,BEKTZ05}.
Necessary and sufficient conditions for unistochasticity are known for $N=3$
\cite{Poon,JS88, Na96}, while the constraints for unistochasticity recently
obtained by Di{\c t}{\v a} \cite{Di06} for the general case of $N \ge4$ are
formulated implicitly and do not provide a constructive solution of the problem.

Geometrical properties of the set of $\mathcal{U}_{3}$ of unistochastic
matrices of order $3$ were studied in \cite{Poon,BEKTZ05}. This set contains a
$4$-dimensional unistochastic ball centered at the flat (van der Waerden)
matrix $W$, for which all entries are equal to $1/3$. The set $\mathcal{U}%
_{3}$ is not convex but it is star shaped. This means that if $B
\in\mathcal{U}_{3}$ then the entire interval $BW$ belongs to the set. The
volume of $\mathcal{U}_{3}$ with respect to the Euclidean (Lebesgue) measure
was numerically estimated by a Monte Carlo-type procedure \cite{BEKTZ05}.

The main aim of this work is to derive an analytical formula for the volume of
the set $\mathcal{U}_{3}$ of unistochastic matrices of order $3$ with respect
to the flat, Euclidean measure. We obtain also a compact expression allowing
one to average any function of elements of $B$ over the set $\mathcal{U}_{3}$.
In particular we derive an explicit result for the average generalized
entropies $S_{q}$ which in the special case $q=1$ gives the mean Shannon
entropy of the columns of unistochastic matrices. We compute also the average
Jarlskog invariant $J$ proportional to the area $A$ of the unitarity triangle,
which characterizes any unistochastic matrix \cite{Ja85, JS88}. Furthermore we
derive higher moments $\langle J^{k} \rangle$ and analyze the probability
distribution $P(J)$, to get more insight into properties of the CKM matrix,
which describes violation of the CP symmetry. Such an approach was recently
suggested by Gibbons et al. \cite{GGPT09}, who used other probability measures
for this purpose.

All averages are computed with respect to the natural Euclidean measure in the
set of unistochastic matrices and are compared with the averages with respect
to the measure induced by the Haar measure on $U(3)$. This measure leading to
the \textsl{unistochastic ensemble} \cite{ZKSS03}, called flag--manifold
measure in \cite{GGPT09}, is not uniform in the set $\mathcal{U}_{3}$.
Although approach presented, based on properties of intertwining operators
associated to the group $S_{3}$ \cite{D1,D2}, is directly applicable to the
case $N=3$, we conjecture also some properties of the measures in the sets of
unistochastic matrices in higher dimensions.

The paper is organized as follows. In the next two sections we present the
necessary definitions and review key properties of the Birkhoff polytope
$\mathcal{B}_{3}$ and its subset $\mathcal{U}_{3}$ containing unistochastic
matrices. In section 4 we define a family of measures in the set
$\mathcal{U}_{3}$ and compute its volume with respect to them. Similar study
of average entropies is presented is section 4 while the average Jarlskog
invariant and its distribution are investigated in section 5. The paper is
concluded with section 6, while some conjectures concerning the measures in
the set of unitary matrices for $N\ge4$ are relegated to the Appendix.


\section{The Birkhoff polytope $\mathcal{B}_{3}$}

A real square matrix $B$ of order $N$ is called \textsl{bistochastic} (or
doubly stochastic)  if it satisfies the following conditions%

\begin{equation}
\mathrm{{i)}\ }B_{ij}\geq0\mathrm{\hspace{10mm}{ii)}\ }\sum_{i}B_{ij}%
=1\mathrm{\hspace{10mm}{iii)}\ }\sum_{j}B_{ij}=1\mathrm{\ .}\label{bist}%
\end{equation}

Such a matrix $B$ is often used to describe discrete dynamics, $p^{\prime}%
=Bp$, in the space of probability vectors. Condition i) implies that all
elements of the transformed vector $p^{\prime}$ are non-negative. Due to
condition ii) its $1$--norm $\sum_{i} p_{i}$, is preserved. A matrix
satisfying two first conditions is called \textsl{stochastic} and it sends the
simplex of $N$-point probability vectors into itself. Condition iii) implies
that additionally the transposed matrix $B^{T}$ is stochastic, which explains
the name.

The uniform probability vector $p_{*}$ with all components equal, $p_{i}=1/N$
stays clearly invariant with respect to any bistochastic matrix, $Bp_{*}%
=p_{*}$. Thus a bistochastic matrix describes a (weak) contraction of the
probability simplex towards the uniform distribution $p_{*}$.

Let $\mathcal{B}_{N}$ denote the set of all bistochastic matrices of order
$N$, called \textsl{Birkhoff polytope}. This convex polytope is well known in
linear programming. Since it arises in the problem of assigning $N$ workers to
$N$ tasks, given their efficiency ratings for each task,  it is sometimes
called the \textsl{assignment polytope}.

The Birkhoff polytope is equivalent to the convex hull of all $N!$ permutation
matrices of size $N$. Hence a permutation matrix $P$ forms an extremal point
of $\mathcal{B}_{N}$. All corners of $\mathcal{B}_{N}$ are equivalent in the
sense that a given corner can be obtained from another one by an orthogonal
transformation. A bistochastic matrix belongs to the boundary of the Birkhoff
polytope if and only if at least one of its entries is equal to zero.

There exists a unique bistochastic matrix $W$, with all entries equal, $W_{ij}
= \frac{1}{N}$. It is also called a matrix of van der Waerden, since it
saturates the van der Waerden inequality \cite{MO79} concerning the permanent
of bistochastic matrices, per$(B)\ge N! N^{-N}$. It is easy to see that $W$ is
located symmetrically at the center of the Birkhoff polytope.

A bistochastic matrix $B$ can be determined by its minor of size $(N-1)$.
Hence the dimensionality of the Birkhoff polytope $\mathcal{B}_{N}$ equals
$(N-1)^{2}$. For instance, the dimension of the set $\mathcal{B}_{2}$ is equal
to one, and this set forms indeed an interval between the identity matrix
${\mathbbm 1}_{2}$ and the $2$--element permutation matrix. In other words,
any bistochastic matrix of order two can be written as
\begin{equation}
B_{2} \left(  a\right)  = \left[
\begin{array}
[c]{cc}%
a & 1-a\\
1-a & a
\end{array}
\right]  \mathrm{\quad where \quad} a \in[0,1].\label{bist2}%
\end{equation}

The length of this interval, equivalent to volume of $\mathcal{B}_{2}$, is
equal to unity, if we consider it as a subset of ${\mathbbm R}^{1}$. However,
we are going to consider this set as an element of ${\mathbbm R}^{N^{2}}$ then
the distance between the points $(1,0,1,0)$ and $(0,1,0,1)$ is equal to $2$,
so in these units one has $\mathrm{vol}_{1}(\mathcal{B}_{2})=2$.

In this paper we are going to work with the case $N=3$, so the Birkhoff
polytope is defined as a convex hull of $3!=6$ permutation matrices,
\begin{equation}
P=\left[
\begin{array}
[c]{ccc}%
0 & 1 & 0\\
0 & 0 & 1\\
1 & 0 & 0
\end{array}
\right]  \ , \quad\quad P^{2}=P^{-1}=\left[
\begin{array}
[c]{ccc}%
0 & 0 & 1\\
1 & 0 & 0\\
0 & 1 & 0
\end{array}
\right]  \ , \quad\quad{\mathbbm 1}=\left[
\begin{array}
[c]{ccc}%
1 & 0 & 0\\
0 & 1 & 0\\
0 & 0 & 1
\end{array}
\right]  \ , \quad\quad\label{paa}%
\end{equation}

\begin{equation}
P_{12}=\left[
\begin{array}
[c]{ccc}%
0 & 1 & 0\\
1 & 0 & 0\\
0 & 0 & 1
\end{array}
\right]  \ , \quad\quad P_{13}=\left[
\begin{array}
[c]{ccc}%
0 & 0 & 1\\
0 & 1 & 0\\
1 & 0 & 0
\end{array}
\right]  \ , \quad\quad P_{23}=\left[
\begin{array}
[c]{ccc}%
1 & 0 & 0\\
0 & 0 & 1\\
0 & 1 & 0
\end{array}
\right]  \ .\label{pbb}%
\end{equation}

We divided the permutation matrices into two triples, which belong to two
totally orthogonal $2$--planes. A uniform mixture in any triple produces the
flat matrix $W$,
\begin{equation}
P + P^{2} + {\mathbbm 1} \ = \ W \ = \ P_{12}+P_{13} + P_{23}
\ .\label{triangles}%
\end{equation}
Working with the standard Hilbert-Schmidt distance, defined by $D^{2}%
(A,B):=\mathrm{Tr}(A-B)(A-B)^{*}$, we see that both triples form equilateral
triangles. To produce a sketch of them in $4$ dimensions we will use the
following parametrization
\begin{equation}
B({\vec b})=B(b_{1},b_{2},b_{3},b_{4}) \ := \ \left[
\begin{array}
[c]{ccc}%
b_{1} & b_{2} & 1-b_{1}-b_{2}\\
b_{3} & b_{4} & 1-b_{3}-b_{4}\\
1-b_{1}-b_{3} & 1-b_{2}-b_{4} & \sum_{i=1}^{4}b_{i}-1
\end{array}
\right]  .\label{Bbbb}%
\end{equation}

\begin{figure}[ptbh]
\centerline{ \hbox{
\epsfig{figure=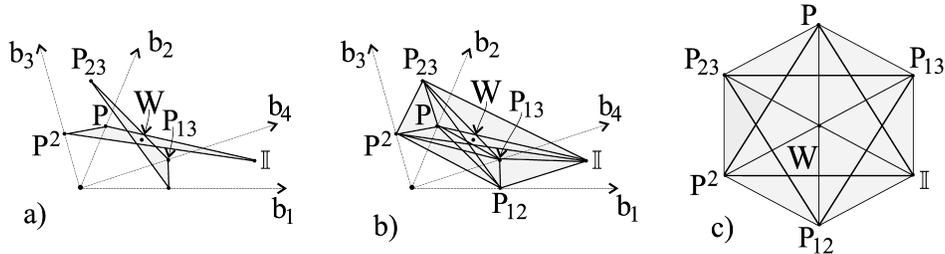,width=12.5cm}
}}  \caption{Birkhoff polytope for $N = 3$ plotted in parametrization
(\ref{Bbbb}), a) two orthogonal equilateral triangles centered at the flat van
der Waerden matrix $W$, b) all $15$ edges determining the polytope, c) the
same polytope as seen 'from above'.}%
\label{fig1}%
\end{figure}

Both triangles shown in Fig. \ref{fig1}a, cross at their center $W$. Six
permutation matrices form $15$ edges, out of which all belong to the boundary
of $\mathcal{B}_{3}$ and all are extremal. There are six long edges, of length
$\sqrt{6}$, which form two equilateral triangles, and nine short edges of
length $2$. If one plots all of them, as in Fig. \ref{fig1}b, the sketch of
$\mathcal{B}_{3}$ becomes complete, but not very illuminating. Another natural
possibility is to look at the polytope 'from above', the direction
distinguished by the vector ${\vec b}=(1/3,1/3,1/3,1/3)$. The Birkhoff
polytope then appears symmetrically as a regular hexagon with two inscribed
equilateral triangles forming the Star of David - see Fig. \ref{fig1}c. Note
that all the diagonals of the hexagon belong to the boundary of $\mathcal{B}%
_{3}$, and that the distance $PP_{12}$ of the diagonal is shorter than the
side $P{\mathbbm 1}$ of the equilateral triangle. Any $2$-d face of the
polytope is formed by an isosceles triangle with two short edges and one long.

The polytope $\mathcal{B}_{3}$ is defined as the convex hull of $6$ corners,
so one could think, it may be decomposed into two $4$-d simplices, each
determined by $5$ points. This would be possible, if we could select $4$
corners, which span the base of a simplex and then allow two other corners to
play the role of an apex for two simplices, with the same base. However, this
would require that the edge connecting both apexes is not extremal or it
includes one of the corner from the base of the simplex. Neither of these
holds for the Birkhoff polytope, so its decomposition into two simplices is
not possible.

To find a decomposition of $\mathcal{B}_{3}$ into three simplices take three
corners of one equilateral triangle, e.g. $\triangle( P,P^{2},{\mathbbm 1})$.
Out of the orthogonal triangle select a side, say the one formed by the
corners $P_{12}$ and $P_{13}$. These five corners define a $4$-d simplex. The
same construction performed for two other sides of the $\triangle
(P_{12},P_{13},P_{23})$ produces two other simplices. It is easy to show any
point of $\mathcal{B}_{3}$ belongs to one of these simplices and that the
$4$-d volume of any of their intersections is equal to zero. Such a
triangulation of the Birkhoff polytope allows to find that its volume  in
${\mathbb{R}}^{9}$ according to the Lebesgue measure is equal to $9/8$. A
detailed investigation of the geometry of the Birkhoff polytope is provided in
\cite{BG77}.


\section{The set $\mathcal{U}_{3}$ of unistochastic matrices}

\label{sec:unist}

A certain class of bistochastic matrices can be generated from unitary
matrices. Let $U$ denote a unitary matrix. Unitarity condition, $UU^{*}%
={\mathbbm 1}$, implies that the matrix $B$ defined by
\begin{equation}
B=f(U), \mathrm{\quad so \quad that \quad} B_{ij} = |U_{ij}|^{2}\label{unist}%
\end{equation}
is bistochastic. Any bistochastic matrix $B\in\mathcal{B}_{N}$ for which there
exists unitary $U \in U(N)$ such that $B=f(U)$ is called
\textsl{unistochastic}. The set of all unistochastic matrices of size $N$ will
be denoted as $\mathcal{U}_{N}$.

Note that the multiplication of $U$ by any diagonal unitary matrices $D_{1}$
and $D_{2}$ changes the phases of entries of $U$, but does not modify the
corresponding bistochastic matrix. Hence we define an equivalence relation
\begin{equation}
U \approx U^{\prime}\ = \ D_{1} U D_{2}\label{equiv1}%
\end{equation}
and observe that $B=f(U)=f(U^{\prime})$.

If the unitary matrix $U$, appearing in (\ref{unist}) is orthogonal, the
corresponding bistochastic matrix $B$ is called \textsl{orthostochastic}%
.\footnote{In some papers this name is used  for unistochastic matrices as
well.} This is the case for any bistochastic matrix of size $2$, since writing
an orthogonal matrix $O =\left(
\begin{smallmatrix}
\cos\vartheta & \sin\vartheta\\
-\sin\vartheta & \cos\vartheta
\end{smallmatrix}
\right) $ and taking $\vartheta= \arccos a$ we see that $B=f(O)$ for any
bistochastic matrix of order two represented in the form (\ref{bist2})
Therefore any unistochastic matrix of order $2$ is also orthostochastic. This
is no longer the case for $N\ge3$, as it is explicitly shown later in this section.

Interestingly this simple mathematical observation has far reaching
consequences for physics. In the theory of elementary particles one defines a
discrete space--time symmetry called CP, which stands for \textsl{charge}
conjugation and \textsl{parity}. Such a symmetry requires that a physical
process in which all particles are replaced by their antiparticles is
equivalent to the mirror image of the original process.

If such a symmetry were obeyed the CKM matrix $V_{\mathrm{CKM}}$ would be
orthogonal, (and thus invariant with respect to the complex conjugation) or it
would be equivalent to an orthogonal matrix with respect to (\ref{equiv1}).
For $N=2$ this is the case for any unitary matrix from $U(2)$. As the CP
symmetry was discovered in 1964 to be violated in experiments on decay of
neutral mesons $K$, one could predict that the number $N$ of generations of
quarks in the theory, equal to the size of the CKM matrix, has to be greater
than two.

In her first paper on the CKM matrix \cite{Ja85} Jarlskog observed that for
any unitary matrix $U$ of size $3$ the number
\begin{equation}
J := \mathrm{Im} ( U_{11}U_{22} U^{*}_{12}U^{*}_{21} )\label{ja1}%
\end{equation}
is invariant with respect to multiplication of the matrix $U$ by diagonal
unitary matrices and permutations. This quantity, now called the
\textsl{Jarlskog invariant}, computed for the CKM matrix $V_{\mathrm{CKM}}$
can be considered as a measure of the violation of the CP symmetry.

Consider now an arbitrary bistochastic matrix $B$ of size $N$. To check if
this matrix is unistochastic we need to know whether there exists a unitary
matrix $U$ such that $f(U)=B$ according to eq. (\ref{unist}). The moduli of
the unitary matrix are determined by the square roots of the entries of the
bistochastic matrix, $|U_{ij}|=\sqrt{B_{ij}}$ and one needs to find a set of
phases $\phi_{ij}$ which guarantees unitarity, where $U_{ij}=\left\vert
U_{ij}\right\vert \exp\left(  \mathrm{i}\phi_{ij}\right)$.

Choose the two first columns of $U$, which we denote as $\left\vert
u_{1}\right\rangle $ and $\left\vert u_{2}\right\rangle $. Their orthogonality
relation, $\left\langle u_{1}|u_{2}\right\rangle =0$, implies that
\[
\sum_{j=1}^{N}U_{j1}U_{j2}^{\ast}=0 .
\]
Introducing the notation $L_{j}=\left\vert U_{j1}\right\vert \cdot\left\vert
U_{j2}\right\vert $ and $\theta_{j}=\phi_{j1}-\phi_{j2}$ we may rewrite this
relation as
 \begin{equation} 
 \sum_{j=1}^{N}L_{j}\exp\left(  \mathrm{i}\theta_{j}\right)  =0.
\label{link1} 
\end{equation}

This form has a nice geometric interpretation: given a set of $N$ line
segments of lengths $L_{1},\dots L_{N}$ we need to find the phases $\theta
_{j}$ in such a way that the entire chain is closed. Obviously it cannot be
done unless the longest link is shorter or is equal to the sum of all other
links. We are free to change the order of summation in (\ref{link1}) and hence
to relabel the links in such a way that they are ordered non increasingly,
$L_{1} \ge L_{2} \ge\dots\ge L_{N}$. Then the chain links condition reads
\begin{equation}
L_{1} \ \le\ L_{2} + \dots+ L_{N} \ .\label{link2}%
\end{equation}
If it is satisfied the chain can be closed and forms a unitarity polygon.

This relation was imposed by the assumed orthogonality of the first two
columns of $U$, but analogous conditions should be fulfilled by all links
corresponding to any pair of columns of $U$. Similar conditions are due to the
orthogonality between any two rows of $U$. This implies the total number of
$N(N-1)$ constraints of the form (\ref{link2}), some of which can be dependent
\cite{PZK01,ZKSS03}. However, there is an example of a bistochastic matrix $B$
of order four, which satisfies all chain links condition for all pair of rows,
but not for all pair of columns \cite{Pa02,ZKSS03}, so in practice one has to
check rows and columns separately. Furthermore, for $N\ge4$ these conditions
for $B$ are only necessary but not sufficient to imply unistochasticity.

\begin{figure}[ptbh]
\centerline{ \hbox{
\epsfig{figure=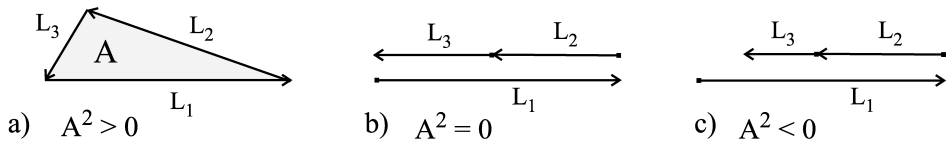,width=12.5cm}
}}  \caption{Chain link condition for unistochasticity: a) a unistochastic
matrix with a positive area of the unitarity traingle, $A^{2}>0$, b) limiting
case: an orthostochastic matrix with $A^{2}=0$, c) a bistochastic matrix $B$
not included into $\mathcal{U}_{3}$ for which $A^{2} < 0$.}%
\label{fig2}%
\end{figure}

It is comforting to realize that the situation gets simpler for $N=3$. In this
case the chain links relation for the first two columns reduces to the
triangle inequality,
\begin{equation}
|L_{2}-L_{3}| \ \le\ L_{1} \ \le\ L_{2} + L_{3} \ ,\label{link33}%
\end{equation}
and the first constraint is required if we relax the assumption that the links
are ordered decreasingly. Although in general one should check similar
relations stemming to other pairs of columns and rows of $U$, in the case
$N=3$ the last column by construction has the right moduli and does not impose
any further restrictions \cite{Na96}. Thus in this case
the relation allowing a chain to close is sufficient for unistochasticity  \cite{Poon,JS88},
and explicit formulae  for the phases $\phi_{ij}$ are provided below. 
If there is a bistochastic matrix $B$ such that in all relations (\ref{link33})
equality takes place, the phases $\theta_{j}$ are equal to zero or to $\pi$.
Thus the corresponding matrix $U$ is orthogonal, which means that $B$ is
orthostochastic. It is easy to show that a matrix $B$ belongs to the boundary
of the set $\mathcal{U}_{3}$ if and only if $B$ is orthostochastic. The set
$\mathcal{U}_{3}$ forms a $4$--dimensional subset of $\mathcal{B}_{3}$ of a
positive measure, while the set of orthostochastic matrices, at the boundary
of $\mathcal{U}_{3}$, is three dimensional \cite{Poon,BEKTZ05}.

For any given $B$ of order three it is straightforward to check whether link
conditions (\ref{link33}) are fulfilled, so that $B$ is unistochastic. For
instance nine short edges, (of length $2$) of the Birkhoff polytope belong to
$\mathcal{U}_{3}$, while the long edges (of length $\sqrt{6}$) do not belong
to this set.

Let us take three such edges spanned by $P,P^{2}$ and ${\mathbbm 1}$, which
form the equilateral triangle. At this plane, the set of orthostochastic
matrices, for which $L_{1}= L_{2} + L_{3}$, forms a \textsl{deltoid} -- see
Fig. \ref{fig3}a. This figure also called $3$--\textsl{hypocycloid}, may be
obtained by sliding a circle of radius $1/3$ inside the unit circle. Thus the
set of unistochastic matrices corresponds to the interior of the deltoid and
is not convex. This set contains the maximal unistochastic ball of radius
$r=\sqrt{2}/3$ centered at $W$, which touches the boundary at the deltoid.

Incidentally, the very same figure is related in a different way with the set
of unistochastic matrices of order three. The spectra of these matrices are
real or belong to the deltoid inscribed into the unit disk and stemming from
the real eigenvalue equal to unity \cite{ZKSS03}.

\begin{figure}[ptbh]
\centerline{ \hbox{
\epsfig{figure=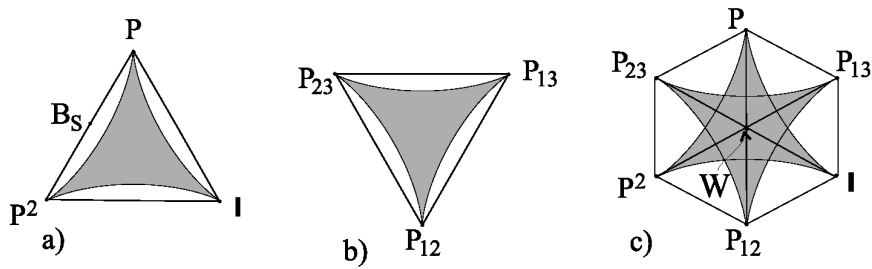,width=11.5cm}
}}  \caption{Non convex set $\mathcal{U}_{3}$ of unistochastic matrices forms
a proper subset $\mathcal{B}_{3}$. a) Deltoid obtained by the cross-section of
$\mathcal{U}_{3}$ along the plane spanned by the equilateral triangle
$\triangle(P,P^{2},{\mathbbm 1})$, b) a similar cross-section along totally
orthogonal plane, c) a view 'from above' as in \ref{fig2}c.}%
\label{fig3}%
\end{figure}

Consider a unistochastic matrix $B$ parametrized by (\ref{Bbbb}). The length
of the links read
\begin{equation}
L_{1} = \sqrt{b_{1}b_{2}}, \quad L_{2} = \sqrt{b_{3}b_{4}}, \quad L_{3} =
\sqrt{(1-b_{1}-b_{2})(1-b_{3}-b_{4})},\label{link123}%
\end{equation}
and the triangle inequality (\ref{link33}) provides the direct condition for
unistochasticity. Let us write down the area $A$ of this \textsl{unitarity
triangle} with sides $L_{1},L_{2}$, and $L_{3}$, and semiperimeter $p =
(L_{1}+L_{2}+L_{3})/2$. Making use of the Heron's formula
\begin{equation}
A \ = \ \sqrt{p(p-L_{1})(p-L_{2})(p-L_{3})} \ ,\label{area}%
\end{equation}
and substituting (\ref{link123}) we arrive with a compact expression for the
squared area $A^{2}$.

It will be convenient to work with this quantity multiplied by sixteen,
\begin{equation}
Q(b) := 4b_{1}b_{2}b_{3}b_{4}-\left(  b_{1}+b_{2}+b_{3}+b_{4}-1-b_{1}%
b_{4}-b_{2}b_{3}\right)  ^{2} = 16A^{2}.\label{area2}%
\end{equation}

Here $b=\{b_{1},b_{2},b_{3},b_{4}\}$ represents a vector in $\mathbb{R}^{4}$
which determines a bistochastic matrix in parametrization (\ref{Bbbb}). In
fact we can form six unitarity triangles in this way, depending on what pair
of columns or rows we wish to choose. Although their shapes differ due to
unitarity their area $A$ is the same \cite{JS88} so the quantity $A$ does not
change under permutation of the unitary matrix $U$.

It is easy to see that all chain--links conditions are equivalent to the
single condition for unistochasticity,
\begin{equation}
A^{2} (B) \ \geq\ 0 \ .\label{areacond}%
\end{equation}
In other words $B\left(  b\right)  \in\mathcal{U}_{3}$ if and only if
$b\in\Omega$ where
\begin{equation}
\Omega:=\left\{  b\in\mathbb{R}^{4}:b_{1}\geq0,\ b_{2}\geq0, \ b_{1}+b_{2}%
\leq1, \ Q\left(  b\right)  \geq0\right\}  .
\label{conduni3}%
\end{equation}

Also $\Omega$ is the closure of the connected component of $\left\{
b:Q\left(  b\right)  >0\right\}  $ which contains $\left(  \frac{1}{3}%
,\frac{1}{3},\frac{1}{3},\frac{1}{3}\right)  $ (see \cite[Sect. 2]{D1}).

We relate these expressions to one of the standard parametrizations of a
unitary matrix (see Di\c{t}\u{a} \cite[p. 11]{Di06}). For any $U^{\prime}\in
U\left(  3\right)  $ there are diagonal matrices $D_{1},D_{2}\in U\left(
3\right)  $ such that $D_{1}U^{\prime}D_{2}=$%
\begin{equation}
U=%
\begin{bmatrix}
c_{12} & s_{12}c_{13} & s_{12}s_{13}\\
s_{12}c_{23} & -c_{12}c_{13}c_{23}-e^{\mathrm{i}\delta}s_{13}s_{23} &
e^{\mathrm{i}\delta}c_{13}s_{23}-c_{12}c_{23}s_{13}\\
s_{12}s_{23} & e^{\mathrm{i}\delta}c_{23}s_{13}-c_{12}c_{13}s_{23} &
-c_{12}s_{13}s_{23}-e^{\mathrm{i}\delta}c_{13}c_{23}%
\end{bmatrix}
. \label{unitar3}%
\end{equation}

The parameters are the angles $\theta_{12},\theta_{13},\theta_{23},\delta$ and
$c_{jk}:=\cos\theta_{jk},s_{jk}:=\sin\theta_{jk},0\leq\theta_{jk}\leq\frac
{\pi}{2}$ for $1\leq j<k\leq3$. The function (\ref{unist}) determines thus a
unistochastic matrix $B\left(  b\right)  =f\left(  U\right)  $ and its entries read%

\begin{align}
b_{1} &  =c_{12}^{2},\ \ b_{2}=s_{12}^{2}c_{13}^{2},\ \ b_{3}=s_{12}^{2}c_{23}%
^{2},\label{unitar4}\\
b_{4} &  =c_{12}^{2}c_{13}^{2}c_{23}^{2}+s_{13}^{2}s_{23}^{2}+2c_{12}%
c_{13}c_{23}s_{13}s_{23}\cos\delta.\nonumber
\end{align}
We explain how these parameters are used to determine the phases $\phi_{ij}$
for $2\leq i,j\leq3$. We consider only the nondegenerate case in which all
entries of $U$ are nonzero. This implies $c_{jk}>0$ and $s_{jk}>0$ for $1\leq
j<k\leq3$. The equations (\ref{unitar4}) determine $\delta$ up to the sign (in
$-\pi<\delta<\pi$). The fact that $U$ and $\overline{U}:=\left[
\overline{U_{ij}}\right]  _{i,j=1}^{3}$ produce the same values for
$b_{1},\ldots,b_{4}$ causes this ambiguity. We will adopt the
normalization\ $\operatorname{Im}U_{22}>0$. This forces $\sin\delta
<0,\operatorname{Im}U_{32}<0,\operatorname{Im}U_{23}<0$ and $\operatorname{Im}%
U_{33}>0$. The cosines of the phases are computed using the entries in
(\ref{unitar3}). The phases are related to the (interior) angles of the
unitarity triangles derived from columns $1$ and $2$, and from columns $1$ and
$3$. For the former case, using the lengths from equation (\ref{link123}) and
denoting the angles $\theta_{1},\theta_{2},\theta_{3}$ by the label on the
opposite side, we find%
\begin{align*}
\cos\theta_{3} &  =\frac{L_{1}^{2}+L_{2}^{2}-L_{3}^{2}}{2L_{1}L_{2}}%
=\frac{b_{1}b_{2}+b_{3}b_{4}-\left(  1-b_{1}-b_{3}\right)  \left(
1-b_{2}-b_{4}\right)  }{2\sqrt{b_{1}b_{2}b_{3}b_{4}}}\\
&  =-\cos\phi_{22},\\
\cos\theta_{2} &  =\frac{L_{1}^{2}+L_{3}^{2}-L_{2}^{2}}{2L_{1}L_{3}}%
=\frac{b_{1}b_{2}+\left(  1-b_{1}-b_{3}\right)  \left(  1-b_{2}-b_{4}\right)
-b_{3}b_{4}}{2\sqrt{b_{1}b_{2}\left(  1-b_{1}-b_{3}\right)  \left(
1-b_{2}-b_{4}\right)  }}\\
&  =-\cos\phi_{32}.
\end{align*}
From the conditions $\sin\phi_{22}>0$ and $\sin\phi_{32}<0$ we obtain
$\phi_{22}=\pi-\theta_{3}$ and $\phi_{32}=\theta_{2}-\pi$ (note the interior
angles satisfy $0<\theta_{1},\theta_{2},\theta_{3}<\pi$ thus $0<\phi_{22}<\pi$
and $-\pi<\phi_{32}<0$). By interchanging columns $2$ and $3$ we find the
remaining nonzero phases (we use matrix entry notation from equation
(\ref{Bbbb}) for more concise statements):%
\begin{align*}
\cos\phi_{23} &  =-\frac{b_{1}B_{13}+b_{3}B_{23}-B_{31}B_{33}}{2\sqrt
{b_{1}b_{3}B_{13}B_{23}}},\ \ \sin\phi_{23}<0;\\
\cos\phi_{33} &  =-\frac{b_{1}B_{13}+B_{31}B_{33}-b_{3}B_{23}}{2\sqrt
{b_{1}B_{13}B_{31}B_{33}}},\ \ \sin\phi_{33}>0.
\end{align*}
For the example where each $b_{i}=\frac{1}{3}$ the unitarity triangle is
equilateral, each $\theta_{i}=\frac{\pi}{3}$ and the above equations give
$\phi_{22}=\frac{2\pi}{3}=\phi_{33}$ and $\phi_{32}=-\frac{2\pi}{3}=\phi_{23}$.

We return to the consideration of the Jarlskog invariant. A straightforward
computation yields
\begin{equation}
2\operatorname{Re}\left(  U_{11}U_{22}U_{12}^{\ast}U_{21}^{\ast}\right)
=1-b_{1}-b_{2}-b_{3}-b_{4}+b_{1}b_{4}+b_{2}b_{3}, \label{2rej}%
\end{equation}
and thus the square of the Jarlskog invariant (\ref{ja1}) reads
\begin{equation}
\bigl[ \operatorname{Im}\left(  U_{11}U_{22}U_{12}^{\ast}U_{21}^{\ast}\right)
\bigr]^{2} =b_{1}b_{2}b_{3}b_{4}- \bigl[ \operatorname{Re}\left(  U_{11}%
U_{22}U_{12}^{\ast} U_{21}^{\ast}\right)  \bigr]^{2}\label{conduni4}%
\end{equation}

Substituting expression (\ref{2rej}) into above equation and comparing the
outcome with (\ref{area2}) we see that
\begin{equation}
J^{2} \ = \ \frac{1}{4}Q(b) \ = \ 4 A^{2}\label{ja2}%
\end{equation}
Thus the squared Jarlskog invariant, proportional to the squared area of the
unitarity triangle, may also be defined as in (\ref{area2}) for an arbitrary
bistochastic matrix $B\in\mathcal{B}_{3}$. For simplicity we shall write
according to the context $J=J(U)$ or $J=J(B)=J \bigl( B(f(U) \bigr)$, as it
should not lead to misunderstanding.

The squared Jarlskog invariant, $J^{2}$, is equal to zero if and only if $B$
is orthostochastic, so there exists an orthogonal matrix $O$, such that
$B_{ij}=O_{ij}^{2}$.

Following Haagerup \cite{Ha96} we shall call two unitary matrices $U_{1}$ and
$U_{2}$  \textsl{equivalent}, written $U_{1} \sim U_{2}$, if there exist two
diagonal unitary matrices $D_{1}$ and $D_{2}$ and two permutation matrices
$P_{1}$ and $P_{2}$ such that
\begin{equation}
U_{2} \ = \ D_{1} P_{1} \cdot U_{1} \cdot P_{2} D_{2} \ .
\label{equival}%
\end{equation}
Observe that due to permutation matrices this relation is more general than
the relation (\ref{equiv1}).

Since multiplication by phases or permutations do not vary the area of the
unitarity triangle we see that the squared Jarlskog invariant of two
equivalent unitaries are equal, if $U_{1} \approx U_{2}$ then $J^{2}%
(U_{1})=J^{2}(U_{2})$. Going back to the set of unitary matrices $U(3)$ we see
that $J(U)=0$ if $U$ is orthogonal, or more generally, if $U$ is
\textsl{equivalent} to an orthogonal matrix, $U= \ D_{1} P_{1} \cdot O \cdot
P_{2} D_{2}$.

Thus $J^{2}(U)$ measures to what extend the matrix $U$ can be transformed into
an orthogonal matrix by means of enphasing and permutations. It is easy to see
that $J$ is maximal if the unitarity triangle is equilateral, $L_{1}%
=L_{2}=L_{3}=1/3$ so that $J^{2}_{\mathrm{max}}=1/108$
-- see Appendix \ref{sec:ext}.
This is the case for
the flat matrix $W$ of van der Waerden, which corresponds to the unitary
Fourier matrix
\begin{equation}
F_{3} = \frac{1}{\sqrt{3}} \left[
\begin{array}
[c]{ccc}%
1 & 1 & 1\\
1 & \exp( \mathrm{i} \cdot2\pi/3) & \exp( \mathrm{i} \cdot4\pi/3)\\
1 & \exp( \mathrm{i} \cdot4\pi/3) & \exp( \mathrm{i} \cdot2\pi/3)
\end{array}
\right]  \ ,\label{f3}%
\end{equation}
which is an example of a complex Hadamard matrix of order three.
Such a unitary matrix
$H$ of size $N$ is distinguished by an extra condition that all its complex
entries have the same modulus, 
$|H_{ij}|^{2}=1/N$ for $i,j=1,\dots N$ \cite{TZ06}.

Any complex Hadamard matrix $H$ of order three is known to be
equivalent to the Fourier matrix $F_3$ \cite{Ha96}, 
which implies that $J^{2}(H)=J^{2}_{\mathrm{max}}$.
From this perspective the set of Hadamard matrices is
maximally distant from the set of orthogonal matrices.
While complex Hadamard matrices correspond to the flat bistochastic matrix $W$, 
located at the center of the set of unistochastic matrices, 
the boundary of which is formed by the set of orthostochastic matrices.

On the other hand for any bistochastic matrix which is not unistochastic, the
quantity $J^{2}=Q/4$ defined by (\ref{area2}) is negative. Since $Q=16A^{2}$
one might say that in this case the area of the unitarity triangle is
imaginary, since the three segments $L_{i}$ cannot be closed to form a
triangle. Among all bistochastic matrices of size three the quantity $Q$ is
the smallest for the matrix $B_{S}$ of Schur (\ref{B3}) for which $Q=-\frac
{1}{16}$, see Appendix \ref{sec:ext}. Indeed, looking at Fig. \ref{fig3}a we
see that $B_{S}$ is such a point of the Birkhoff polytope $\mathcal{B}_{3}$,
for which the distance to the set $\mathcal{U}_{3}$ of unistochastic matrices,
represented by the gray deltoid, is maximal.


\section{The volume of the set of unistochastic matrices of order three.}

\label{sec:vol}

Since unistochastic matrices of order three are used in various branches of
theoretical physics, several geometric properties of the set $\mathcal{U}_{3}$
of these matrices were studied in \cite{BEKTZ05}. In particular, in that work
the volume of this set was estimated numerically. In this section we shall
improve these findings by deriving an analytical formula for this volume. To
this end we need to introduce probability measures into the set of
unistochastic matrices. A first natural choice will be

\medskip\noindent a) the flat (Lebesgue) measure $\mu_{3/2}$ used before in
\cite{BEKTZ05}.

\noindent

The above notation is due to the fact that this measure belongs to a
one-parameter class of measures $\mu_{k}$, defined below. In general any
probability measure in the set $U(N)$ of unitary matrices induces by function
(\ref{unist}) a measure into the set $\mathcal{U}_{N}$. Thus we will
distinguish the case

\medskip\noindent b) the measure $\mu_{1}$ induced by the Haar measure on
$U(3)$.

\noindent

This measure leads to the \textsl{unistochastic ensemble} or random
unistochastic matrices \cite{ZKSS03}. Since it is related to the unitarily
invariant measure on the flag manifold $U(3)/[U(1)]^{3}$ it was called
\textsl{flag--manifold measure} in \cite{GGPT09}.

\medskip

Suppose now that $dm$ denotes the normalized Haar measure on $U\left(
3\right)  $, $db$ is Lebesgue measure on $\mathbb{R}^{4}$, and $g$  is a
continuous function on $\left[  0,1\right]  ^{4}\subset\mathbb{R}^{4}$. Due to
\cite[Theorem 2.1]{D1} one may relate the integrals over the space of unitary
matrices $U(3)$ and over the set $\Omega\in\mathbb{R}^{4}$, see eq.
(\ref{conduni3}), which determines the set of unistochastic matrices,%

\begin{equation}
\int_{U(3)} \! g \! \left(  \left\vert U_{11}\right\vert ^{2}, \left\vert
U_{12}\right\vert ^{2},\left\vert U_{21}\right\vert ^{2}, \left\vert
U_{22}\right\vert ^{2}\right)  dm =\frac{2}{\pi}\int_{\Omega} \! g \! \left(
b_{1},b_{2},b_{3},b_{4}\right)  Q\left(  b\right)  ^{-1/2}db.\label{integrals}%
\end{equation}
The function $g(b_{1},b_{2},b_{3},b_{4})=g(b)$ determines the function $g(B)$
defined on the entire set $\mathcal{U}_{3}$ of unistochastic matrices.

We introduce new coordinates $\left(  b_{1},s,t,r\right)  $ with
\begin{align}
b_{2}  &  =s\left(  1-b_{1}\right)  , \quad\quad\quad b_{3} =t\left(
1-b_{1}\right)  ,\nonumber\\
b_{4}  &  =\left(  1-s\right)  \left(  1-t\right)  +b_{1}st+2r \sqrt
{b_{1}st\left(  1-s\right)  \left(  1-t\right) } . \label{defb4}%
\end{align}
Note that $s=c_{23}^{2},t=c_{13}^{2},r=\cos\delta$ in terms of the matrix $U$.

Then $Q=4b_{1}\left(  1-b_{1}\right)  ^{2}s\left(  1-s\right)  t\left(
1-t\right)  \left(  1-r^{2}\right)  $, and $\Omega$ corresponds to%
\[
\left\{  \left(  b_{1},s,t,r\right)  :\left(  b_{1},s,t\right)  \in\left[
0,1\right]  ^{3},-1\leq r\leq1\right\}  .
\]

The Jacobian for the change-of-variables is
\[
\left\vert \frac{\partial\left(  b_{1},b_{2},b_{3},b_{4}\right)  }%
{\partial\left(  b_{1},s,t,r\right)  }\right\vert =2\left(  1-b_{1}\right)
^{2}  \sqrt{ b_{1}s\left(  1-s\right)  t\left(  1-t\right)  } .
\]

Making use of the expressions derived in \cite{D1} we may write an explicit
form for an integral of a continuous function $g$ with respect to the measure
$Q(b)^{k-3/2}$ for arbitrary $k > 1/2$. Hence it is natural to introduce a
one-parameter family of measures $\mu_{k}$ on $\mathcal{U}_{3}$ which satisfy
\begin{equation}
\int_{\mathcal{U}_{3}}g(B)d\mu_{k}=    \int_{\Omega}g\left(  b\right)
Q(b)^{k-3/2}db. 
\label{mukint}%
\end{equation}

Thus the case $k=3/2$ corresponds to the flat (Lebesgue) measure on
$\mathcal{U}_{3}$, while for $k=1$ this expression reduces to (\ref{integrals}%
), so $\mu_{1}$ represents the measure induced  by the Haar measure on $U(3)$.

Integral at the right hand side of (\ref{mukint}) can be rewritten as
\begin{gather}
\int_{\Omega}g\left(  b\right)  Q\left(  b\right)  ^{k-\frac{3}{2}}db=\int
_{0}^{1}\int_{0}^{1}\int_{0}^{1}\int_{-1}^{1}g\left(  b_{1},s\left(
1-b_{1}\right)  ,t\left(  1-b_{1}\right)  ,b_{4}\right) \label{intJ}\\
\times b_{1}^{k-1}\left(  1-b_{1}\right)  ^{2k-1}\left[  4s\left(  1-s\right)
t\left(  1-t\right)  \right]  ^{k-1}\left(  1-r^{2}\right)  ^{k-\frac{3}{2}%
}dr~ds~dt~db_{1},\nonumber
\end{gather}
where $b_{4}$ is given by (\ref{defb4}). Setting $g=1$ and using the standard
beta integrals we determine the normalization constant \cite[Prop. 3.2]{D1})%
\begin{equation}
h_{k}:=\int_{\Omega}Q\left(  b\right)  ^{k-\frac{3}{2}}db=\frac{\pi
\Gamma\left(  k\right)  ^{3}}{\left(  2k-1\right)  \Gamma\left(  3k\right)
},\quad k>\frac{1}{2}.\label{hk}%
\end{equation}

Suppose $n=1,2,3,\ldots$: then
\begin{align*}
h_{n}  &  =\frac{\pi\left(  \left(  n-1\right)  !\right)  ^{3}}{\left(
2n-1\right)  \left(  3n-1\right)  !},\\
h_{n+1/2}  &  =\frac{\pi^{2}}{2n}\frac{ \bigl[ \left(  \frac{1}{2}\right)
_{n}\bigr]^{2}} {\left(  n+\frac{1}{2}\right)  _{2n+1}},
\end{align*}
where $\left(  x\right)  _{n}:=\prod_{i=1}^{n}\left(  x+i-1\right)  $ stands
for the Pochhammer symbol.  In particular,%
\begin{equation}
h_{3/2}=\dfrac{\pi^{2}}{3\cdot5\cdot7}=\dfrac{\pi^{2}}{105} .\label{h32}%
\end{equation}
gives the volume of $\mathcal{U}_{3}$ considered as a subset of $\mathbb{R}%
^{4}$. This is multiplied by $9$ to produce the volume relative to
$\mathbb{R}^{9}$ \cite{CM07} -- see Appendix \ref{sec:ext}.

Thus the ratio of the $4$--dimensional volume of $\mathcal{U}_{3}$ in
$\mathbb{R}^{9}$ to the volume of all bistochastic matrices is
\begin{equation}
\frac{\mbox{vol}(\mathcal{U}_{3})} {\mbox{vol}(\mathcal{B}_{3})} =
9\times\frac{\pi^{2}}{105}/\left(  \frac{9}{8}\right)  =\frac{8\pi^{2}}%
{105}=0.751969...
\end{equation}
This is in agreement with the outcome of earlier numerical calculations
\cite[eqn. (24)]{BEKTZ05} which were based on roughly $10^{7}$ sample points
and yielded $0.7520\pm0.0005$.

For completeness let us add that the volume of the set $\mathcal{U}_{3}$ with
respect to the flag--manifold measure $\mu_{1}$ reads $h_{1}=\pi/2$.


\section{Mean entropy of a $N=3$ unistochastic matrix.}

The Shannon entropy of an $N$-point probability vector $p=(p_{1},\dots p_{N})$
is defined by
\begin{equation}
S (p) = - \sum_{i=1}^{N} p_{i}\ln p_{i} \ .\label{entrdef0}%
\end{equation}
This quantity measures to what extent the vector is mixed and varies from $0$
for any pure vector $(1,0,\dots0)$ to $\ln N$ for the maximally mixed vector
$p_{*}=(1/N, \dots1/N)$. In an analogous way one defines the entropy of a
bistochastic matrix
\begin{equation}
S(B) = - \frac{1}{N}\sum_{i=1}^{N} \sum_{j=1}^{N} B_{ij}\ln{B_{ij}}
\ ,\label{entrdef}%
\end{equation}
equal to the average entropy of its rows (or columns). For any permutation
matrix $P$ this entropy is equal to zero while its maximum value $\ln{N}$ is
attained at the flat matrix $W$. The mean entropy of a bistochastic matrix was
considered by S{\l }omczy{\'n}ski \cite{Sl02} and later analyzed in
\cite{ZKSS03,GMP03}.

To derive an expression for the average entropy of a unistochastic matrix with
respect to any probability measure on $\mathcal{U}_{3}$ bi-invariant under the
symmetric group $\mathcal{S}_{3}$ we observe  that is equal to the expected
value of $-3B_{11}\ln B_{11}$. This can be computed with respect to the
probability measure $h_{k}^{-1}Q\left(  b\right)  ^{k-\frac{3}{2}}db$ on
$\Omega$.

Let us denote the mean entropy by $\left\langle S\right\rangle _{\Omega,k}$,
where parameter $k$ labels the measure defined in (\ref{mukint}). From
(\ref{intJ}) specialized to functions of $b_{1}$ we obtain
\begin{align}
\left\langle S\right\rangle _{\Omega,k}  &  =\frac{-3\Gamma\left(  3k\right) }
{\Gamma\left(  k\right)  \Gamma\left(  2k\right)  }\int_{0}^{1}b_{1}^{k}\ln
b_{1}\left(  1-b_{1}\right)  ^{2k-1}db_{1}\nonumber\\
&  =\psi\left(  3k+1\right)  -\psi\left(  k+1\right)  ,\label{entrpsi1}%
\end{align}
where the digamma function reads  $\psi\left(  x\right)  :=\Gamma^{\prime
}\left( x\right)  /\Gamma\left(  x\right)  ,x>0$.

The recurrence relation $\psi\left( x+1\right)  =\psi\left(  x\right)  + 1/x $
is used in the following: Suppose $n=1,2,3,\cdots$ then
\begin{equation}
\left\langle S\right\rangle _{\Omega,n} =\sum_{j=n+1}^{3n}\frac{1}{j},
\quad\quad\left\langle S\right\rangle _{\Omega,n+\frac{1}{2}} =\sum
_{j=n+1}^{3n+1}\frac{2}{2j+1}.
\end{equation}
In particular the average entropy for the flag-manifold measure $\mu_{1}$
gives $\left\langle S\right\rangle _{\Omega,1}=\frac{5}{6}=0.833...$, This
result coincides with the average entropy of random complex vectors
\cite{Jo90,BZ06} which form a unitary matrix.

The mean entropy with respect to the flat measure $\mu_{3/2}$ on
$\mathcal{U}_{3}$ reads  $\left\langle S\right\rangle _{\Omega,\frac{3}{2}%
}=2\left(  \frac{1}{5}+\frac{1}{7}+ \frac{1}{9}\right)  =\frac{286}%
{315}=0.90793651\ldots$. This quantity was approximated as $0.908$ in
\cite[eqn. (25)]{BEKTZ05}.

These data can be compared with the maximal entropy $S_{\mathrm{max}} = \ln{3}
\approx1.099$, characteristic of the flat matrix $W$ of van der Waerden.

For comparison let us now compute the mean entropy with respect to the
Lebesgue measure averaged over the entire set $\mathcal{B}_{3}$ of
bistochastic matrices.

Straightforward calculation allows us to integrate functions of $b_{1},b_{2}$
over $\mathcal{B}_{3}$ with respect to the flat measure $db:=\prod_{i=1}%
^{4}db_{i}$. In general, we consider an arbitrary function $f$, integrable on
the triangle with vertices $\left\{  \left(  0,0\right)  ,\left(  1,0\right)
,\left(  0,1\right)  \right\}  $.
The result is
\begin{equation}
\int_{\mathcal{B}_{3}}f\left(  b_{1},b_{2}\right)  db=\int_{0}^{1}db_{1}%
\int_{0}^{1-b_{1}}f\left(  b_{1},b_{2}\right)  \left(  b_{1}b_{2}+\left(
b_{1}+b_{2}\right)  \left(  1-b_{1}-b_{2}\right)  \right)  db_{2}%
.\label{formbb}%
\end{equation}
The corollary stated in the Appendix specializes this to $\int_{0}^{1}f\left(
b_{1}\right)  db$ and allows us to find the volume of the Birkhoff polytope.
Thus relative to $\mathbb{R}^{9}$ one has $\mathrm{vol}_{4}\left(
\mathcal{B}_{3}\right)  =9/8$.

Formula (\ref{formbb})  can also be used to compute the average entropy, equal
to the expected value of $(-3b_{1}\ln b_{1})$. The analytic result
\[
8\int_{\mathcal{B}_{3}}\left(  -3b_{1}\ln b_{1}\right)  db=\frac{53}{60}.
\]
agrees with the numerical estimation $0.883$  obtained earlier in
\cite{BEKTZ05}. Note that this number is smaller than the mean entropy
$\left\langle S\right\rangle _{\Omega,\frac{3}{2}}$  averaged over the set of
unistochastic matrices, since these bistochastic matrices which do not belong
to $\mathcal{U}_{3}$ are located close to the boundary of the Birkhoff
polytope and are characterized by small entropy.

\medskip

Let us now turn to the generalized  entropy defined for any probability vector
$\left\{  p_{i}:1\leq i\leq N\right\} $
\begin{equation}
S_{q}:=\frac{1}{q-1}\sum_{i=1}^{N}\left(  p_{i}-p_{i}^{q}\right)
.\label{renyi}%
\end{equation}

The parameter $q\neq1$ is assumed to be non negative. In the limiting case the
generalized entropy converges to the standard (Shannon) entropy,
$\lim_{q\rightarrow1}S_{q}=-\sum_{i=1}^{N}p_{i}\ln p_{i}$.

Applying definition (\ref{renyi}) to a unistochastic matrix $B\in
\mathcal{U}_{3}$ similarly to the ordinary case we let
\[
S_{q}=\frac{1}{3\left(  q-1\right)  }\sum_{i=1}^{3}\sum_{j=1}^{3} \left(
B_{ij}-B_{ij}^{q}\right)  .
\]
We compute the expected value of this expression with respect to $h_{k}%
^{-1}Q\left(  b\right)  ^{k-\frac{3}{2}}db$ on $\Omega$ for $k>\frac{1}{2}$.

It is the same as the expected value of $\frac{3}{q-1}\left(  B_{11}%
-B_{11}^{q}\right)  $, indeed%
\begin{align*}
\left\langle S_{q}\right\rangle _{\Omega,k}  &  =\frac{3\Gamma\left(
3k\right)  }{\left(  q-1\right)  \Gamma\left(  k\right)  \Gamma\left(
2k\right)  }\int_{0}^{1}\left(  b_{1}-b_{1}^{q}\right)  b_{1}^{k-1}\left(
1-b_{1}\right)  ^{2k-1}db_{1}\\
&  =\frac{1}{q-1}\left(  1-\frac{3\Gamma\left(  k+q\right)  \Gamma\left(
3k\right)  }{\Gamma\left(  k\right)  \Gamma\left(  3k+q\right)  }\right)  .
\end{align*}
When $k$ is an integer number $n$ or a half-integer $k=n+\frac{1}{2}$ this
expression is a rational function of $q$ and can be expressed with help of the
Pochhammer symbol $\left(  x\right) _{n}$ defined above,
\begin{align}
\left\langle S_{q}\right\rangle _{\Omega,n}  &  = \frac{1}{q-1}\left(
1-\frac{\left(  3n\right)  !}{n!\left(  q+n\right)  _{2n}}\right)  .\\
\left\langle S_{q}\right\rangle _{\Omega,n+\frac{1}{2}}  &  =  \frac{1}
{q-1}\left(  1-\frac{3\left(  \frac{1}{2}+n\right)  _{2n+1}}{\left(
q+\frac{1}{2}+n\right)  _{2n+1}}\right) \label{rengen}%
\end{align}

In particular, taking $n=1$ we arrive at handy expressions for the mean
generealized entropies averaged over the Haar measure and flat measure
respectively, which allow for an explicit partial fraction expansion,

\begin{eqnarray} 
\left\langle S_{q}\right\rangle _{\Omega,1}  &  =  \dfrac{q+4}{\left(
q+1\right)  \left(  q+2\right)  } =\frac{3}{q+1}-\frac{2}{q+2},\\
\left\langle S_{q}\right\rangle _{\Omega,\frac{3}{2}}  &  \; = \;
\dfrac{2\left(  4q^{2}+34q+105\right)  }{\left(  2q+3\right)  \left(
2q+5\right)  \left(  2q+7\right) } = \dfrac{63}{4\left(  2q+3\right)  }%
-\dfrac{45}{2\left(  2q+5\right)  }+\dfrac{35}{4\left(  2q+7\right)  } .
\nonumber
\label{renmeas}
\end{eqnarray}

For completeness we provide also an expression for the generalized entropy
averaged over the set $\mathcal{B}_{3}$ with respect to the flat measure
obtained with help of Corollary \ref{col:int}
\begin{equation}
\left\langle S_{q}\right\rangle _{\mathcal{B}_{3}} =\frac{2}{q+1}+\frac
{4}{q+2}-\frac{9}{q+3}+\frac{4}{q+4}.\label{renbist}%
\end{equation}

These entropies characterize well the distribution of matrices generated by
these measures. In particular, a comparison of both expressions in
(\ref{renmeas})  shows that the Haar measure on $U(3)$ populates the region
close to the boundary of  $\mathcal{U}_{3}$ more densely then the vicinity of
the flat matrix $W$ around its center. Since the squared Jarlskog invariant
$J^{2}$ is by construction equal to zero at the boundary of $\mathcal{U}_{3}$,
we may expect that its mean value over the flat measure $\mu_{1}$ is smaller
than the average with respect to the Haar measure $\mu_{3/2}$. As shown in the
next section this is indeed the case.

\section{Distribution of the Jarlskog invariant}

\label{sec:dist}

The value of the Jarlskog invariant and its square at a cross-section of the
set $\mathcal{U}_{3}$ of unistochastic matrices is shown in Fig. \ref{fig5}.
Recent papers of Gibbons et al. \cite{GGPT08,GGPT09} analyzed squared Jarlskog
invariant \cite{Ja85,JS88} averaged over several probability measures on the
set of unitary matrices. In particular these authors computed the expectation
value, $\langle J^{2}\rangle$, averaged over the 'flag manifold' measure
induced by the Haar measure on $U(3)$ and analyzed numerically the probability
distribution $P(|J|)$ with respect to this measure. In this section we proceed
one step further and derive an analytical formula for this probability distribution.

\begin{figure}[ptbh]
\centerline{ \hbox{
\epsfig{figure=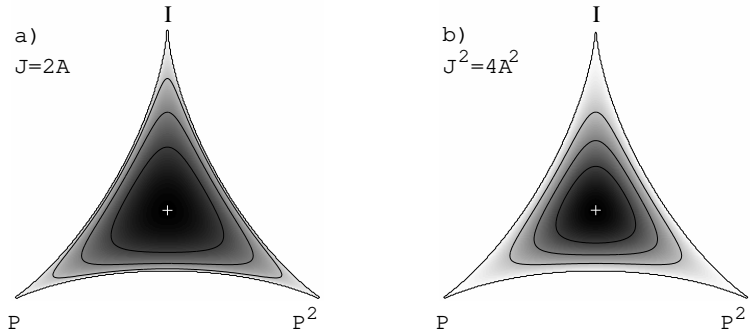,width=10cm}
}}  \caption{The absolute value of the Jarlskog invariant $|J|$ a), and its
square $J^{2}$ b), at the cross-section of $\mathcal{U}_{3}$ along the plane
formed by two permutation matrices and the identity. Dark color denotes high
values of $|J|$ and $J^{2}$. The maximum is achieved at the van der Waerden
matrix $W$ located the center of the deltoid. Note that outside the deltoid
$J^{2} <0$ and the bistochastic matrix is not unistochastic.}%
\label{fig5}%
\end{figure}

We shall start computing the moments of the distribution of the variable
$Q=4J^{2}$ defined in (\ref{area2}) as a function of a random unistochastic
matrix $B$. This task is rather simple, since we can express the moments of
$Q$ with respect to any measure $\mu_{k}$ by the coefficients $h_{k}$ defined
in (\ref{hk}),
\begin{align}
\langle Q^{n}\rangle_{k} =\frac{h_{k+n}}{h_{k}}  &  =3^{-3n}\frac{\left(
k-\frac{1}{2}\right)  \Gamma\left(  k+\frac{1}{3}\right)  \Gamma\left(
k+\frac{2}{3}\right)  \Gamma\left(  k+n\right) ^{2}}{\left(  k+n-\frac{1}%
{2}\right)  \Gamma\left(  k+n+\frac{1}{3}\right)  \Gamma\left(  k+n+\frac
{2}{3}\right)  \Gamma\left(  k\right)  ^{2}}\nonumber\\
&  =3^{-3n}\frac{\left(  k-\frac{1}{2}\right)  \left(  k\right)  _{n}^{2}%
}{\left(  k+n-\frac{1}{2}\right)  \left(  k+\frac{1}{3}\right)  _{n}\left(
k+\frac{2}{3}\right)  _{n}},\label{momq}%
\end{align}
where $k$ determines the measure (\ref{mukint}) while $n=0,1,2,\ldots$.

Setting $k=1$ and $n=1$ we find that the mean squared Jarlskog invariant,
averaged over the Haar measure reads  $\langle J^{2}\rangle_{1}=\langle
Q/4\rangle_{1}=1/720 = 1.389 \times10^{-3}$  in consistence with \cite[eqn.
(75)]{GGPT09}. For comparison note that the average over the flat measure
yields a larger value, $\langle J^{2}\rangle_{3/2}=3/1144=2.622 \times10^{-3}%
$. In general the flat measure favors larger values of $|J|$ as it is shown in
Fig.\ref{fig6}.

Having at our disposal the complete set of the moments of $Q$ we will
determine the exact distribution function $P(Q)$  in terms of hypergeometric
and related functions. To avoid nuisance factors in the calculations we will
consider the random variable $X:=27Q=108 J^{2}$ so that  $X$ takes values in
$\left[  0,1\right] $.

We know that the following relations hold for any $k> 1/2$ and $n=0,1,2,\cdots
$
\begin{align*}
\left(  k-\frac{1}{2}\right)  \int_{0}^{1}x_{0}^{n}x_{0}^{k-\frac{3}{2}%
}dx_{0}  &  =\frac{\left(  k-\frac{1}{2}\right)  }{\left(  k+n-\frac{1}%
{2}\right)  },\\
\frac{\Gamma\left(  \alpha+\beta\right)  }{\Gamma\left(  \alpha\right)
\Gamma\left(  \beta\right)  }\int_{0}^{1}x^{n+\alpha-1}\left(  1-x\right)
^{\beta-1}dx  &  =\frac{\left(  \alpha\right)  _{n}}{\left(  \alpha
+\beta\right)  _{n}},\left(  \alpha,\beta>0\right)  .
\end{align*}

In view of the expression (\ref{momq}) for the moments the above relations
allow us to find an alternative representation of the desired probability
distribution $P(X)$.

Let $X_{0},X_{1},X_{2}$ be independent random variables with the densities%
\begin{eqnarray}
\label{xx0}
f_{0}\left(  x\right)   &  = &\left(  k-\frac{1}{2}\right)  x^{k-\frac{3}{2}},\\
\label{xx1}
f_{1}\left(  x\right)   &  = &\frac{\Gamma\left(  k+\frac{1}{3}\right)  }%
{\Gamma\left(  k\right)  \Gamma\left(  \frac{1}{3}\right)  }x^{k-1}\left(
1-x\right)  ^{-2/3},\\
f_{2}\left(  x\right)   &  = & \frac{\Gamma\left(  k+\frac{2}{3}\right)  }%
{\Gamma\left(  k\right)  \Gamma\left(  \frac{2}{3}\right)  }x^{k-1}\left(
1-x\right)  ^{-1/3},
\label{xx2}
\end{eqnarray}
respectively, each being defined on $0\leq x\leq1$.

Then $X$ has the same moments and the same probability distribution as the
product $X_{0}X_{1}X_{2}$. This step, justified in Appendix \ref{sec:jarprob}
enables us to arrive at the key result of this section: an explicit expression
for the probability distribution for $X=108 J^{2}$, where $J$ denotes the
Jarlskog invariant of a random unistochastic matrix generated according to the
measure $\mu_{k}$,
\begin{equation}
P_{k}\left(  X \right)  =c_{k}\left(  k-\frac{1}{2}\right)  X^{k-3/2}\int_{X}
^{1}~_{2}F_{1}\left(  \frac{1}{3},\frac{2}{3};1;1-t\right)  t^{-1/2}%
dt.\label{int2F1}%
\end{equation}
It is assumed here that $0< X \leq1$ and $k> 1/2$, the symbol $_{2}F_{1}$
stands for the hypergeometric function, while the normalization constant
reads
\begin{equation}
c_{k}:=\dfrac{\Gamma\left(  k+\frac{1}{3}\right)  \Gamma\left(  k+\frac{2}%
{3}\right)  }{\Gamma\left(  k\right)  ^{2}} ,\label{ck}%
\end{equation}

In Appendix \ref{sec:jarprob} we made use of this expression to determine an
explicit expansion for the density function $P_{k}(X)$. These results allowed
us to show the distributions $P_{1}(X)$ and $P_{3/2}(X)$ in Fig. \ref{fig6}.
Observe that in the case $k=1$ we obtain an expression for the distribution of
$\left\vert J\right\vert $, considered as a function of a random unitary
matrix $U$ distributed with respect to the Haar measure on $U(3)$ random
variable on For small values of $|J|$ this distribution behaves as $P(|J|)
\sim\alpha|J|^{\lambda}$ with $\lambda=0$ and $\alpha= 8 \pi\approx25.133$.

\begin{figure}[ptbh]
\centerline{ \hbox{
\epsfig{figure=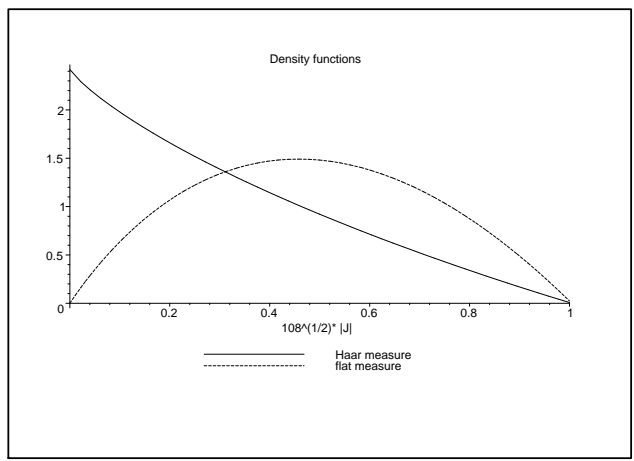,width=6cm,angle=-90}
}}  \caption{Probability distribution $P_{k}(|J|)$ of the Jarlskog invariant
for random unistochastic matrices generated according to the Haar measure
$\mu_{1}$ and the flat measure $\mu_{3/2}$.}%
\label{fig6}%
\end{figure}

To make a direct connection with the results of \cite{GGPT09}. we reproduce
here the formula for the integrated probability distribution. For any
$0<y\leq\frac{1}{6\sqrt{3}}$ the probability of finding a unitary matrix $U$
distributed according to the Haar measure on $U(3)$ with $|J(B)|$ less or
equal $y$ reads,
\begin{equation}
P_{1}\left\{  \left\vert J\right\vert \leq y\right\}  = 8\pi y+\left\{
24\ln\left(  4y^{2}\right)  \left(  y^{2}+4y^{4}+96y^{6}+\cdots\right)
-72y^{2}+128y^{4}+\frac{24384}{5}y^{6}+\cdots\right\}  .\label{intmu1}%
\end{equation}

In Appendix \ref{sec:jarprob} we derive this formula obtained for the measure
$\mu_{1}$ on $\mathcal{U}_{3}$ as well as an analogous result for the flat
measure $\mu_{3/2}$
\begin{equation}
P_{3/2}\left\{  \left\vert J\right\vert \leq y\right\}  =420y^{2}-\frac
{4480}{\pi}y^{3}\cdots+\frac{1680}{\pi}y^{3}\ln\left(  4y^{2}\right)
+\cdots\label{intmu3}%
\end{equation}
Both cumulative distribution functions are compared in Fig. \ref{fig7}.

\begin{figure}[ptbh]
\centerline{ \hbox{
\epsfig{figure=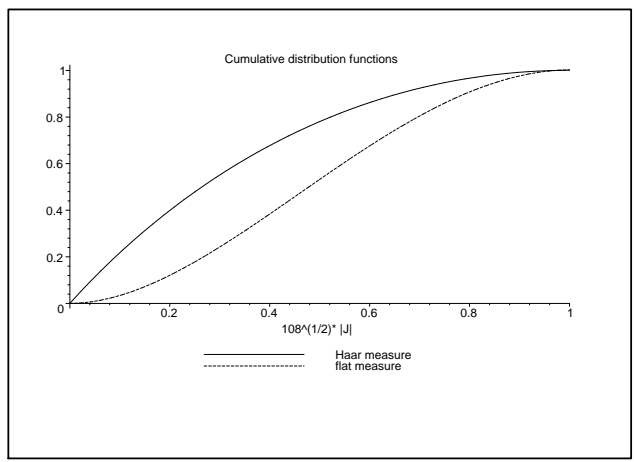,width=6cm,angle=-90}
}}  \caption{Cumulative probability distribution for the functions plotted in
Fig. \ref{fig6}.}%
\label{fig7}%
\end{figure}

Recent experimental data show that the observed value of the Jarlskog
invariant reads \cite{CKMf}
\begin{equation}
J_{\mathrm{obs}} \ = \ 3.08\; [+0.16, -0.18]\; \times10^{-5}\label{Jobser}%
\end{equation}
This concrete number can be now compared with the probability distribution
(\ref{intmu1}).  A $95\%$ confidence interval for $|J|$ is about
$[0.00202,\;1/6\sqrt{3}\approx0.096]$,  while the probability of getting a
value of $|J|$ outside this interval is $5\%$.

Moreover, we get an explicit estimate for the probability of obtaining at
random a unitary matrix,
such that the absolute values of its Jarlskog invariant is smaller then the
observed value,
\begin{equation}
P\left\{  \left\vert J\right\vert \leq J_{\mathrm{obs}}\right\}  =7.74\left[
+0.41,-0.45\right]  \times10^{-4}.\label{prob1}%
\end{equation}

The statistical hypothesis that the CKM matrix arises from the probability
experiment of producing a random unitary matrix, with respect to Haar measure
on $U\left(  3\right)  $, is rejected at the descriptive significance level of
$0.08\%$.

Another benchmark introduced in Gibbons et al \cite{GGPT09},
$P_{\mathrm{flag}}(|J| \le10^{-4}) \approx2.5085 \times10^{-3}$ is consistent
with numerical data obtained in Eq.(92) of that paper.
Note that Gibbons et al \cite{GGPT09} constructed several probability models
for which values as small or smaller than $J_{\mathrm{obs}}$ are more likely.

For comparison we note that the flat measure $\mu_{3/2}$ in the set of
unistochastic matrices yields a smaller probability. Using this measure
($k=\frac{3}{2}$) we obtain $P_{3/2}\left\{  \left\vert J\right\vert \leq
J_{\mathrm{obs}}\right\}  \approx3.98\times10^{-7}$. Indeed this could be
viewed as statistical evidence that the transition probabilities in the CKM
matrix do arise from a unitary matrix. Specifically the so-called likelihood
ratio test applied to the two probability densities for $\left\vert
J\right\vert $ induced by $d\mu_{1}$ and $\frac{8\pi^{2}}{105}d\mu_{3/2}$ (the
factor $\frac{8\pi^{2}}{105}$ comes from $P\left\{  Q<0\right\}  =1-\frac
{8\pi^{2}}{105}$ for the flat measure on $\mathcal{B}_{3}$) at $\left\vert
J\right\vert =3.08\times10^{-5}$ results in a factor of about $1200.$ This
value is obtained from the formulas for $f_{0}\left(  x\right)  $ in Appendix
\ref{sec:jarprob}.

Let us express the Jarlskog invariant (\ref{ja1}) by the standard parameters
(\ref{unitar3}) of a unitary matrix of order three,
\begin{equation}
J(U)=-c_{12}c_{23}c_{13}s_{12}^{2}s_{23}s_{13}\sin\delta\label{ja3}%
\end{equation}
Observed value of the Jarlskog invariant for the CKM matrix does not imply
that the CP violating phase $\delta_{CKM}$ had to be small. In fact
$\delta_{CKM}\in\lbrack62^{\circ},100^{\circ}]$ so even the value $\pi/2$ is
not ruled out -- see e.g. \cite{KN09}. Hence the small value (\ref{Jobser}) is
due to the angles $\theta_{ij}$ in (\ref{unitar3}), which determine the
bistochastic matrix.

Thus we are going to conclude this section with a simple statistical
statement: The CKM matrix should not be considered as a generic unitary matrix
drawn at random with respect to the Haar measure on $U(3)$. Furthermore, the
matrix $B_{CKM}$ of squared entries of $V_{CKM}$ is rather unlikely to be an
ordinary bistochastic matrix generated at random with respect to the flat
measure in this set.


\section{Concluding Remarks}

In this work we have analyzed the Birkhoff polytope $\mathcal{B}_{3}$ of $N=3$
bistochastic matrices and its subset $\mathcal{U}_{3}$ of unistochastic
matrices. This set contains these bistochastic matrices which arise from
squared moduli of entries of a unitary matrix. We have improved the result of
\cite{BEKTZ05} by computing the exact volume of $\mathcal{U}_{3}$ with respect
to the flat (Lebesgue) measure and found that it takes more than three
quarters of the volume of the Birkhoff polytope $\mathcal{B}_{3}$.

We have introduced a one-parameter family of probability measures $\mu_{k}$
into the set $\mathcal{U}_{3}$ of unistochastic matrices. Among the measures
(\ref{mukint}) the distinguished ones are the uniform (flat) measure
$\mu_{3/2}$ and the measure $\mu_{1}$, induced by the Haar measure on $U(3)$.
Furthermore, the measure $\mu_{k}$ obtained in the limit $k\rightarrow\frac
{1}{2}$ coincides with the measure induced by the Haar measure on the
orthogonal group $O\left( 3\right) $.

We derived explicit formulae which allow us to compute expectation values of a
smooth function of an entry of $B$ with respect to these measures. In this way
we derived exact expressions for the mean entropy and the generalized entropy
of a random unistochastic matrix with respect to the measures $\mu_{k}$. These
values can serve as a reference values in studying properties of concrete
unitary and bistochastic matrices of order three, used in the theory of
quantum information.

In high energy physics and the theory of CP symmetry breaking one works with
unitary matrices of order $3$ and characterizes them by the Jarlskog invariant
(\ref{ja1}). Computing all the moments of the squared Jarlskog invariant
$J^{2}$ with respect to the measure $\mu_{k}$ we could represent the
probability distribution $P_{k}(J)$ as an integral (\ref{int2F1}) of the
hypergeometric function $_{2}F_{1}$. Expanding this function in a series and
integrating it term by term we arrived at an explicit representation of the
desired probability distributions. In particular, working with the Haar
measure $\mu_{1}$ we could derive analytical results on the distribution
$P_{1}(|J|)$ consistent with the numerical results earlier obtained in Gibbons
et al. \cite{GGPT09}. Our results support then the observation, that the
unitary CKM matrix $V_{CKM}$,  which describes the violation of the CP
symmetry, should not be regarded as a generic unitary matrix of order $3$.

\vspace{5mm}

\textbf{Acknowledgements:}

\vspace{5mm}

\noindent

It is a pleasure to thank I. Bengtsson and to W. Tadej for numerous
stimulating discussions and helpful correspondence. We acknowledge financial
support by the the special grant number DFG-SFB/38/2007 of Polish Ministry of
Science and Higher Education and an European Research project COCOS (K{\.Z}).

\appendix

\section{Extreme values of the parameter $Q$}

\label{sec:ext}

Consider a bistochastic matrix $M$ of order three parametrized by (\ref{Bbbb})
and the function $Q(M)$ defined in (\ref{area2}). The aim of this section is
to show that for any $M\in\mathcal{B}_{3}$ this function takes values in
$[-1/16, 1/27]$. Note that if the matrix is unistochastic, $M \in
\mathcal{U}_{3}$ than $Q$ is non negative and is proportional to the squared
area of the unitarity triangle.

It is straightforward to show that $Q$ is invariant under transposition and
permutation of rows or columns (for example replacing $\left(  b_{2}%
,b_{4}\right)  $ by $\left(  1-b_{1}-b_{2},1-b_{3}-b_{4}\right)  $). Let us
first introduce parameters $b_{1},s,t,x$ with $0\leq b_{1},s,t\leq1$ and
conditions on $x$ to be determined. Motivated by the unistochastic situation
let%
\begin{align*}
b_{2}  &  =\left(  1-b_{1}\right)  s,\\
b_{3}  &  =\left(  1-b_{1}\right)  t,\\
b_{4}  &  =\left(  1-s\right)  \left(  1-t\right)  +b_{1}st+x.
\end{align*}
Four more conditions must be satisfied for $M\left(  b\right)  $ to be
bistochastic (the inequalities $M\left(  b\right)  _{13}\geq0$ and $M\left(
b\right)  _{31}\geq0$ are already satisfied). The simultaneous inequalities
$M\left(  b\right)  _{23}\geq0$ and $M\left(  b\right)  _{32}\geq0$ are
equivalent to
\begin{align}
x  &  \leq\min\left(  u_{1},u_{2}\right)  ,\label{ineq1}\\
u_{1}  &  :=s\left(  1-t\right)  +b_{1}t\left(  1-s\right) \nonumber\\
u_{2}  &  :=t\left(  1-s\right)  +b_{1}s\left(  1-t\right) \nonumber
\end{align}
and $\left\{  M\left(  b\right)  _{22}\geq0,M\left(  b\right)  _{33}%
\geq0\right\}  $ is equivalent to
\begin{align}
-x  &  \leq\min\left(  \ell_{1},\ell_{2}\right)  ,\label{ineq2}\\
\ell_{1}  &  :=\left(  1-s\right)  \left(  1-t\right)  +b_{1}st,\nonumber\\
\ell_{2}  &  :=st+b_{1}\left(  1-s\right)  \left(  1-t\right)  .\nonumber
\end{align}
With these parameters%
\[
Q\left(  b\right)  =Q^{\prime}\left(  b_{1},s,t,x\right)  :=-\left(
1-b_{1}\right)  ^{2}\left(  x^{2}-4b_{1}st\left(  1-s\right)  \left(
1-t\right)  \right)  .
\]
For fixed $b_{1},s,t$ this is decreasing in $x^{2}$; thus the maximum value
occurs at $x=0$ and then maximizing over $b_{1},s,t$ we obtain $Q\left(
b\right)  =\frac{1}{27}$ when $b_{1}=\frac{1}{3},s=\frac{1}{2}=t$ (that is,
$M\left(  b\right)  _{ij}=\frac{1}{3}$ for $1\leq i,j\leq3$).

Next we show that the minimum value of $Q$ on $\mathcal{B}_{3}$ is $-\frac
{1}{16}$, achieved at the Schur matrix (1): by permutations of rows or
columns, and then transposition of $M$, if necessary, we may assume $0\leq
t\leq s\leq\frac{1}{2}$. In this triangle the bounds $-\ell_{2}\leq x\leq
u_{2}$ apply. Rather than considering $\min\left(  Q^{\prime}\left(
b_{1},s,t,-\ell_{2}\right)  ,Q^{\prime}\left(  b_{1},s,t,u_{2}\right)
\right)  $ in this region we will minimize $Q^{\prime}\left(  b_{1}%
,s,t,-\ell_{2}\right)  $ in the triangle bounded by $t=0,t=s,t=1-s$ (vertices
$\left(  0,0\right)  ,\left(  \frac{1}{2},\frac{1}{2}\right)  ,\left(
1,0\right)  $); this works because $=Q^{\prime}\left(  b_{1},s,t,u_{2}\left(
b_{1},s,t\right)  \right)  =Q^{\prime}\left(  b_{1},1-s,t,-\ell_{2}\left(
b_{1},1-s,t\right)  \right)  $ (writing $\ell_{2},u_{2}$ as functions of
$\left(  b_{1},s,t\right)  $). Then%
\[
Q^{\prime}\left(  b_{1},1-s,t,-\ell_{2}\left(  b_{1},1-s,t\right)  \right)
=-\left(  1-b_{1}\right)  ^{2}\left(  b_{1}\left(  1-s\right)  \left(
1-t\right)  -st\right)  ^{2}.
\]
As a function of $\left(  s,t\right)  $, $Q^{\prime}$ can not have an interior
minimum, so it suffices to check the edges of the triangle. On the edge $t=0$
we have $Q^{\prime}=-\left(  1-b_{1}\right)  ^{2}b_{1}^{2}\left(  1-s\right)
^{2}$ with minimum value of $-\frac{1}{16}$ at $b_{1}=\frac{1}{2},s=0$.

On the edge $t=1-s$ we obtain $Q^{\prime}=-\left(  1-b_{1}\right)  ^{4}%
s^{2}\left(  1-s\right)  ^{2}$ with minimum value of $-\frac{1}{16}$ at
$b_{1}=0,s=\frac{1}{2}=t$ (the Schur matrix for $x=-\frac{1}{4}$).

On the edge $s=t,Q^{\prime}=-\left(  1-b_{1}\right)  ^{2}\left(  b_{1}\left(
1-s\right)  ^{2}-s^{2}\right)  ^{2}$. This function has no interior minimum on
the interval $0\leq s\leq\frac{1}{2}$. The endpoints $s=0$ and $s=\frac{1}{2}$
have already been considered.

There is a neat formula for the integral of functions of$\left(  b_{1}%
,b_{2}\right)  $ over $\mathcal{B}_{3}$ with respect to the flat measure
$db:=\prod_{i=1}^{4}db_{i}$. The derivation of the formula involves adding
over the four regions formed in the unit square by the lines $s=t,s+t=1$.

\begin{proposition}
Let $f$ be integrable on $\left\{  \left(  x,y\right)  :x,y\geq0,x+y\leq
1\right\}  $, then%
\[
\int_{\mathcal{B}_{3}}f\left(  b_{1},b_{2}\right)  db=\int_{0}^{1}db_{1}%
\int_{0}^{1-b_{1}}f\left(  b_{1},b_{2}\right)  \left(  b_{1}b_{2}+\left(
b_{1}+b_{2}\right)  \left(  1-b_{1}-b_{2}\right)  \right)  db_{2}.
\]

\end{proposition}

Observe that the integral kernel is an elementary symmetric function of
$\left(  b_{1},b_{2},1-b_{1}-b_{2}\right)  $.

\begin{corollary}
Let $f$ be integrable on $\left[  0,1\right]  $ then%
\[
\int_{\mathcal{B}_{3}}f\left(  b_{1}\right)  db=\frac{1}{6}\int_{0}%
^{1}f\left(  b_{1}\right)  \left(  1+5b_{1}\right)  \left(  1-b_{1}\right)
^{2}db_{1}.
\]
\label{col:int}
\end{corollary}

Thus $\int_{\mathcal{B}_{3}}1db=\frac{1}{8}$.

When considering $\mathcal{B}_{3}$ as a subset of $\mathbb{R}^{9}$ the element
of volume $db=\prod_{i=1}^{4}db_{i}$ is multiplied by $9$. The map
$B:\mathbb{R}^{4}\rightarrow\mathbb{R}^{9}$, defined in equation \ref{Bbbb} is
affine onto a 4-dimensional linear manifold (translate of a subspace) and its
Jacobian equals $9$, calculated as the square root of the determinant of the
Gram matrix of the images of the unit vectors relative to $B\left(
\overrightarrow{0}\right)  $. For example
\[
b_{1}%
\begin{bmatrix}
1 & 0\\
0 & 0
\end{bmatrix}
\rightarrow%
\begin{bmatrix}
0 & 0 & 1\\
0 & 0 & 1\\
1 & 1 & -1
\end{bmatrix}
+b_{1}%
\begin{bmatrix}
1 & 0 & -1\\
0 & 0 & 0\\
-1 & 0 & 1
\end{bmatrix}
;
\]
the Gram matrix is $%
\begin{bmatrix}
4 & 2 & 2 & 1\\
2 & 4 & 1 & 2\\
2 & 1 & 4 & 2\\
1 & 2 & 2 & 4
\end{bmatrix}
$, and its determinant equals $81$.

In this way we may obtain the average entropy
and directly derive expression (\ref{renbist}) for the generalized entropy.
Furthermore, $8\int_{\mathcal{B}_{3}}Q\left(  b\right)  db=\frac{1}{168}%
,8\int_{\mathcal{B}_{3}}Q\left(  b\right)  ^{2}db=\frac{1}{5940}$, thus the
standard deviation reads $\sigma_{Q}=0.01153\ldots$.

Since the set $\mathcal{U}_{3}$ of unistochastic matrices is the subset of
$\mathcal{B}_{3}$ for which $Q\geq0$, equivalently%
\[
\left\vert x\right\vert \leq R:=\left\{  4b_{1}st\left(  1-s\right)  \left(
1-t\right)  \right\}  ^{1/2},
\]
we see that%
\begin{align*}
\int_{\mathcal{U}_{3}}f\left(  b_{1}\right)  db  &  =\int_{0}^{1}f\left(
b_{1}\right)  \left(  1-b_{1}\right)  ^{2}db_{1}\int_{0}^{1}ds\int_{0}%
^{1}dt\int_{-R}^{R}dx\\
&  =\frac{\pi^{2}}{16}\int_{0}^{1}f\left(  b_{1}\right)  b_{1}^{1/2}\left(
1-b_{1}\right)  ^{2}db_{1},
\end{align*}
and $\int_{\mathcal{U}_{3}}1db=\pi^{2}/105$ in agreement with eq.
(\ref{h32}).  Observe that $R\leq\min\left(  l_{1},l_{2},u_{1},u_{2}\right)  $
by the (well-known) inequality $2\sqrt{xy}\leq x+y$ for $x,y\geq0$; this is
the reason that the integral extends over $0\leq s,t\leq1$.

\section{Jarlskog invariant as a product of three random variables}

\label{sec:jarprob}

In this appendix we show that the product of three random variables
$X_{0}X_{1}X_{2}$ introduced in (\ref{xx0} -- \ref{xx2}) has the same
probability distribution as the rescaled squared Jarlskog invariant
$X:=27Q=108 J^{2}$ of random unistochastic matrices generated with respect to
the measure $\mu_{k}$ defined in (\ref{mukint}). We shall start by quoting the
lemma on probability distribution of a product of two independent random variables

\begin{lemma}
\label{pdfXY} Suppose $Y_{1},Y_{2}$ are random variables on $\left[
0,1\right]  $ with densities $g_{i}$ and c.d.f.'s $G_{i},i=1,2$ (that is,
$G_{i}\left(  x\right)  =\int_{0}^{x}g_{i}\left(  t\right)  dt=P\left\{
Y_{i}\leq x\right\}  $, $0\leq x\leq1$). Then the density for $Y_{1}Y_{2}$ is
$\int_{x}^{1}g_{1}\left(  t\right)  g_{2}\left(  \frac{x}{t}\right)  \frac
{1}{t}dt$.
\end{lemma}

Let us apply this lemma, which can be proved by direct integration, to a
random variable $Y_{2}=X_{0}$ distributed as in (\ref{xx0}).

\begin{corollary}
\label{kpdf} If $g_{2}\left(  t\right)  =\left(  k-\frac{1}{2}\right)
t^{k-\frac{3}{2}}$ with $k>\frac{1}{2}$ then the density of $Y_{1}Y_{2}$ is%
\[
\left(  k-\frac{1}{2}\right)  x^{k-\frac{3}{2}}\int_{x}^{1}g_{1}\left(
t\right)  t^{1/2-k}dt,0<x<1.
\]

\end{corollary}

Making use of the normalization constant $c_{k}$ introduced in (\ref{ck}) we
can write down explicit form for the density of the product $X_{1}X_{2}$.

\begin{proposition}
The density $f_{12}$ of $X_{1}X_{2}$ is $c_{k}x^{k-1}~_{2}F_{1}\left(
\frac{1}{3},\frac{2}{3};1;1-x\right)  ,0<x\leq1$.
\end{proposition}

\begin{proof}
By Lemma \ref{pdfXY}, the density is%
\begin{align*}
f_{12}\left(  x\right)   &  =\frac{c_{k}}{\Gamma\left(  \frac{1}{3}\right)
\Gamma\left(  \frac{2}{3}\right)  }\int_{x}^{1}t^{k-1}\left(  1-t\right)
^{-1/3}\left(  \frac{x}{t}\right)  ^{k-1}\left(  1-\frac{x}{t}\right)
^{-2/3}\frac{dt}{t}\\
&  =\frac{c_{k}x^{k-1}}{\Gamma\left(  \frac{1}{3}\right)  \Gamma\left(
\frac{2}{3}\right)  }\int_{x}^{1}\left(  1-t\right)  ^{-1/3}\left(
t-x\right)  ^{-2/3}t^{-1/3}dt\\
&  =\frac{c_{k}x^{k-1}}{\Gamma\left(  \frac{1}{3}\right)  \Gamma\left(
\frac{2}{3}\right)  }\int_{0}^{1}\left(  1-s\right)  ^{-1/3}s^{-2/3}\left(
x+s\left(  1-x\right)  \right)  ^{-1/3}ds,
\end{align*}
using the substitution $t=x+s\left(  1-x\right)  $. Now restrict $x$ to
$\frac{1}{2}<x\leq1$ then $0\leq\frac{1-x}{x}<1$ and we can expand%
\begin{align*}
\left(  x+s\left(  1-x\right)  \right)  ^{-1/3}  &  =x^{-1/3}\left(
1-s\frac{x-1}{x}\right)  ^{-1/3}\\
&  =x^{-1/3}\sum_{n=0}^{\infty}\frac{\left(  \frac{1}{3}\right)  _{n}}%
{n!}s^{n}\left(  \frac{x-1}{x}\right)  ^{n}.
\end{align*}

Integrating term-by-term we obtain%
\begin{align*}
f_{12}\left(  x\right)   &  =\frac{c_{k}x^{k-1}}{\Gamma\left(  \frac{1}%
{3}\right)  \Gamma\left(  \frac{2}{3}\right)  }\sum_{n=0}^{\infty}%
\frac{\left(  \frac{1}{3}\right)  _{n}\Gamma\left(  \frac{1}{3}+n\right)
\Gamma\left(  \frac{2}{3}\right)  }{n!\Gamma\left(  n+1\right)  }\left(
x-1\right)  ^{n}x^{-n-1/3}\\
&  =c_{k}x^{k-1}\sum_{n=0}^{\infty}\frac{\left(  \frac{1}{3}\right)
_{n}\left(  \frac{1}{3}\right)  _{n}}{n!n!}\left(  -1\right)  ^{n}\sum
_{j=0}^{\infty}\frac{\left(  \frac{1}{3}+n\right)  _{j}}{j!}\left(
1-x\right)  ^{n+j}\\
&  =c_{k}x^{k-1}\sum_{m=0}^{\infty}\frac{\left(  \frac{1}{3}\right)  _{m}}%
{m!}\left(  1-x\right)  ^{m}\sum_{n=0}^{m}\frac{\left(  \frac{1}{3}\right)
_{n}m!\left(  -1\right)  ^{n}}{n!\left(  1\right)  _{n}\left(  m-n\right)  !},
\end{align*}
where the summation variables are changed to $n$ and $m=n+j$. The inner sum is
evaluated with the Chu-Vandermonde sum%

\begin{align*}
\sum_{n=0}^{m}\frac{\left(  \frac{1}{3}\right)  _{n}m!\left(  -1\right)  ^{n}%
}{n!\left(  1\right)  _{n}\left(  m-n\right)  !}  &  =\sum_{n=0}^{m}%
\frac{\left(  \frac{1}{3}\right)  _{n}\left(  -m\right)  _{n}}{n!\left(
1\right)  _{n}}\\
&  =\frac{\left(  1-\frac{1}{3}\right)  _{m}}{\left(  1\right)  _{m}}%
=\frac{\left(  \frac{2}{3}\right)  _{m}}{\left(  1\right)  _{m}}.
\end{align*}

This shows that $f_{12}\left(  x\right)  =c_{k}x^{k-1}~_{2}F_{1}\left(
\frac{1}{3},\frac{2}{3};1;1-x\right)  $ at least for $\frac{1}{2}<x\leq1$, but
both sides are analytic on $0<x<2$ so the equality holds for $0<x\leq1$.
\end{proof}

To derive an expression for the density of the triple product $X=X_{0}%
X_{1}X_{2}$ we need to combine lemma \ref{pdfXY} with corollary \ref{kpdf}.
Hence we can write  $f\left(  x\right)  =\left(  k-\frac{1}{2}\right)
x^{k-3/2}\int_{x}^{1}f_{12}\left(  t\right)  t^{1/2-k}dt$ and $f_{12}\left(
t\right)  =c_{k}t^{k-1}~_{2}F_{1}\left(  \frac{1}{3},\frac{2}{3};1;1-t\right)
$ since $t^{k-1}t^{1/2-k}=t^{-1/2}$. This completes the proof of formula
(\ref{int2F1}) for the distribution of the rescaled squared Jarlskog invariant
$X:=27Q=108 J^{2}$.

Following Gibbons et al \cite{GGPT09} we shall now concentrate on the
probability distribution for the absolute value of the Jarlskog invariant,
$\left\vert J\right\vert $, equal to $\left(  X/108\right)  ^{1/2}$. Let
$f_{0}\left(  x\right)  $ denote the density function of $X^{1/2}$, thus
$f_{0}\left(  x\right)  =2xf\left(  x^{2}\right)  $. It is not hard to compute
a series for $f_{0}\left(  x\right)  $ when $x$ is near $1$. We change the
variable of integration $t=\left(  1-s\right)  ^{2}$ and obtain%
\[
f_{0}\left(  x\right)  =4c_{k}\left(  k-\frac{1}{2}\right)  x^{2k-2}\int
_{0}^{1-x}~_{2}F_{1}\left(  \frac{1}{3},\frac{2}{3};1;s\left(  2-s\right)
\right)  ds.
\]

We expand $\left(  s\left(  2-s\right)  \right)  $ $^{j}$ (for $j=1,2,\ldots
$), collect the coefficients of $s^{m}$ and integrate term-by-term to get
\[
\int_{0}^{1-x}~_{2}F_{1}\left(  \frac{1}{3},\frac{2}{3};1;s\left(  2-s\right)
\right)  ds=\sum_{m=0}^{\infty}\frac{\left(  1-x\right)  ^{m+1}}{m+1}%
\sum_{j=0}^{\left\lfloor m/2\right\rfloor }\frac{\left(  -1\right)
^{j}\left(  \frac{1}{3}\right)  _{m-j}\left(  \frac{2}{3}\right)
_{m-j}2^{m-2j}}{\left(  m-j\right)  !\left(  m-2j\right)  !j!};
\]
there is no nice formula for the inner ($j$-) sum. Thus
\[
f_{0}\left(  x\right)  =4c_{k}\left(  k-\frac{1}{2}\right)  x^{2k-2}\left(
1-x\right)  \left\{  1+\frac{2}{9}\left(  1-x\right)  +\frac{22}{243}\left(
1-x\right)  ^{2}+\frac{310}{6561}\left(  1-x\right)  ^{3}+\cdots\right\}  ,
\]
for $x$ near $1$ (that is, not too close to zero). When $k=1$ we have
$c_{1}=\Gamma\left(  \frac{4}{3}\right)  \Gamma\left(  \frac{5}{3}\right)
=\frac{4\pi\sqrt{3}}{27}$ and%
\[
f_{0}\left(  x\right)  =\frac{8\pi\sqrt{3}}{27}\left(  1-x\right)  \left(
1+\frac{2}{9}\left(  1-x\right)  +\frac{22}{243}\left(  1-x\right)  ^{2}%
+\frac{310}{6561}\left(  1-x\right)  ^{3}+\cdots\right)  .
\]

\begin{lemma}
$\int_{0}^{1}~_{2}F_{1}\left(  \frac{1}{3},\frac{2}{3};1;1-t\right)
t^{-1/2}dt=3$.
\end{lemma}

\begin{proof}
Indeed%
\begin{gather*}
\int_{0}^{1}~_{2}F_{1}\left(  \frac{1}{3},\frac{2}{3};1;1-t\right)
t^{-1/2}dt=\sum_{n=0}^{\infty}\frac{\left(  \frac{1}{3}\right)  _{n}\left(
\frac{2}{3}\right)  _{n}}{n!n!}\int_{0}^{1}\left(  1-t\right)  ^{n}%
t^{-1/2}dt\\
=\sum_{n=0}^{\infty}\frac{\left(  \frac{1}{3}\right)  _{n}\left(  \frac{2}%
{3}\right)  _{n}}{n!n!}\frac{\Gamma\left(  n+1\right)  \Gamma\left(  \frac
{1}{2}\right)  }{\Gamma\left(  n+\frac{3}{2}\right)  }=\frac{\Gamma\left(
\frac{1}{2}\right)  }{\Gamma\left(  \frac{3}{2}\right)  }\sum_{n=0}^{\infty
}\frac{\left(  \frac{1}{3}\right)  _{n}\left(  \frac{2}{3}\right)  _{n}%
}{n!\left(  \frac{3}{2}\right)  _{n}}\\
=\frac{\Gamma\left(  \frac{1}{2}\right)  \Gamma\left(  \frac{3}{2}\right)
\Gamma\left(  \frac{1}{2}\right)  }{\Gamma\left(  \frac{3}{2}\right)
\Gamma\left(  \frac{7}{6}\right)  \Gamma\left(  \frac{5}{6}\right)  }%
=\frac{\pi}{\frac{1}{6}\Gamma\left(  \frac{1}{6}\right)  \Gamma\left(
\frac{5}{6}\right)  }=6\sin\frac{\pi}{6}=3.
\end{gather*}

We used the Gauss sum $_{2}F_{1}\left(  a,b;c;1\right)  =\frac{\Gamma\left(
c\right)  \Gamma\left(  c-a-b\right)  }{\Gamma\left(  c-a\right)
\Gamma\left(  c-b\right)  }$ for $c>a+b$, and the equation $\Gamma\left(
t\right)  \Gamma\left(  1-t\right)  =\frac{\pi}{\sin\pi t}$.
\end{proof}

Thus $f\left(  x\right)  =c_{k}\left(  k-\frac{1}{2}\right)  x^{k-3/2}\left(
3-\int_{0}^{x}~_{2}F_{1}\left(  \frac{1}{3},\frac{2}{3};1;1-t\right)
t^{-1/2}dt\right)  $. To analyze the behavior for $x$ near zero we use the
classical formulas for the hypergeometric series $_{2}F_{1}\left(
a,b;c;t\right)  $ at the singular point $t=1$. The special case $c=a+b$ is
more complicated (see \cite[p. 257]{Le}):%

\begin{align*}
_{2}F_{1}\left(  \frac{1}{3},\frac{2}{3};1;1-t\right)   &  =\frac
{\Gamma\left(  1\right)  }{\Gamma\left(  \frac{1}{3}\right)  \Gamma\left(
\frac{2}{3}\right)  }\sum_{n=0}^{\infty}\frac{\left(  \frac{1}{3}\right)
_{n}\left(  \frac{2}{3}\right)  _{n}}{n!n!}\left(  A_{n}-\ln t\right)
t^{n},\\
A_{n}  &  :=2\psi\left(  n+1\right)  -\psi\left(  n+\frac{1}{3}\right)
-\psi\left(  n+\frac{2}{3}\right)  .
\end{align*}

By the triplication formula for the $\psi$-function (recall $\psi\left(
t\right)  =\frac{d}{dt}\Gamma\left(  t\right)  /\Gamma\left(  t\right)  $),%
\begin{align*}
\psi\left(  n+\frac{1}{3}\right)  +\psi\left(  n+\frac{2}{3}\right)   &
=3\psi\left(  3n\right)  -\psi\left(  n\right)  -3\ln3\\
&  =3\psi\left(  3n+1\right)  -\psi\left(  n+1\right)  -3\ln3,
\end{align*}
because $\psi\left(  t+1\right)  =\psi\left(  t\right)  =\frac{1}{t}$ for
$t>0$; the latter formula is valid for $n\geq0$. Thus%
\begin{align*}
A_{n}  &  =3\left(  \psi\left(  n+1\right)  -\psi\left(  3n+1\right)  \right)
+3\ln3\\
&  =-\sum_{j=n+1}^{3n}\frac{3}{j}+3\ln3.
\end{align*}

Also%
\[
\int_{0}^{x}\left(  A_{n}-\ln t\right)  t^{n-1/2}dt=\frac{2x^{n+1/2}}%
{2n+1}\left(  -\ln x+A_{n}+\frac{2}{2n+1}\right)  .
\]
Thus%
\begin{align*}
f\left(  x\right)   &  =c_{k}\left(  k-\frac{1}{2}\right)  x^{k-3/2}\left\{
3-\frac{\sqrt{3}}{2\pi}\sum_{n=0}^{\infty}\frac{\left(  \frac{1}{3}\right)
_{n}\left(  \frac{2}{3}\right)  _{n}}{n!n!}\frac{2x^{n+1/2}}{2n+1}\left(  -\ln
x+A_{n}+\frac{2}{2n+1}\right)  \right\}  ,\\
f_{0}\left(  x\right)   &  =2c_{k}\left(  k-\frac{1}{2}\right)  x^{2k-2}%
\left\{  3-\frac{\sqrt{3}}{2\pi}\sum_{n=0}^{\infty}\frac{\left(  \frac{1}%
{3}\right)  _{n}\left(  \frac{2}{3}\right)  _{n}}{n!n!}\frac{2x^{2n+1}}%
{2n+1}\left(  -2\ln x+A_{n}+\frac{2}{2n+1}\right)  \right\}  ;
\end{align*}
and we have found the density function of $x=\sqrt{X}=6\sqrt{3}|J|$ exhibiting
the behavior for $x$ near zero.

In the expression for $f_{0}$ the first  few terms inside the braces \{\} are
\[
3-\frac{\sqrt{3}}{2\pi}\left[  -\ln\left(  \frac{x^{2}}{27}\right)  \left(
2x+\frac{4}{27}x^{3}+\frac{4}{81}x^{5}\right)  +4x-\frac{22}{81}x^{3}%
-\frac{49}{405}x^{5}\right]  .
\]
The cumulative distribution function $F_{0}\left(  x\right)  =P\left\{
\sqrt{X}<x\right\}  =P\left\{  |J| < x/6\sqrt{3} \right\} $ is
\begin{align*}
F_{0}\left(  x\right)   &  =2c_{k}\left(  k-\frac{1}{2}\right)  x^{2k-1}\\
&  \times\left\{  \frac{3}{2k-1}-\frac{\sqrt{3}}{2\pi}\sum_{n=0}^{\infty}%
\frac{\left(  \frac{1}{3}\right)  _{n}\left(  \frac{2}{3}\right)  _{n}%
x^{2n+1}}{n!n!\left(  2n+1\right)  \left(  n+k\right)  }\left(  -\ln
\frac{x^{2}}{27}-\sum_{j=n+1}^{3n}\frac{3}{j}+\frac{2}{2n+1}+\frac{1}%
{n+k}\right)  \right\}  .
\end{align*}

The important cases are:

\begin{enumerate}
\item $k=1$, Haar measure on $U\left(  3\right)  $,%
\[
F_{0}\left(  x\right)  =\frac{4\pi\sqrt{3}}{9}x-\frac{2}{9}\left(  -\ln\left(
\frac{x^{2}}{27}\right)  \left(  x^{2}+\frac{1}{27}x^{4}+\frac{2}{243}%
x^{6}+\cdots\right)  +3x^{2}-\frac{4}{81}x^{4}-\frac{127}{7290}x^{6}%
+\cdots\right)  ;
\]

\item $k=\frac{3}{2}$, flat measure,%
\[
F_{0}\left(  x\right)  =\frac{70}{27}x^{2}\left\{  \frac{3}{2}-\frac{\sqrt{3}%
}{2\pi}\left(  -\ln\left(  \frac{x^{2}}{27}\right)  \left(  \frac{2}{3}%
x+\frac{4}{135}x^{3}+\cdots\right)  +\frac{16}{9}x-\frac{86}{2025}x^{3}%
+\cdots\right)  .\right\}
\]

\end{enumerate}

Introducing a new variable, $y=x/6\sqrt{3}$, we arrive at expressions
(\ref{intmu1}) and (\ref{intmu3}) presented in Sec. \ref{sec:dist}.

\section{Conjectures of measures in higher dimensions}

The method used in Section \ref{sec:vol} relies on one of the authors'
\cite{D1} construction of a linear operator commuting with the action of
$\mathcal{S}_{3}$, mapping homogeneous polynomials in three variables to
homogeneous polynomials of the same degree, and depending on a parameter $k$
(a particular case of the \textquotedblleft Dunkl intertwining
operator\textquotedblright). This operator is realized as an integral over
$U\left(  3\right)  $. The case $k=1$ is based on a formula of Harish-Chandra
(see Helgason \cite[p.328]{He})%
\[
\int_{U\left(  N\right)  } \exp\left(  \mathrm{Tr} \left(  D\left(  x\right)
UD\left(  y\right)  U^{\ast}\right)  \right)  dm\left(  U\right)  =\frac
{c_{N}}{a\left(  x\right)  a\left(  y\right)  }\sum_{w\in\mathcal{S}_{N}}%
\det\left(  w\right)  \exp\left(  \left\langle xw,y\right\rangle \right)  ,
\]
where the symmetric group on $N$ objects is identified with the set of
permutation matrices in $O\left(  N\right)  $,  for $x,y\in\mathbb{R}^{N}$ the
inner product is $\left\langle x,y\right\rangle :=\sum_{j=1}^{N}x_{j}y_{j}$,
$D\left(  x\right)  $ is the diagonal matrix with $D\left(  x\right)
_{jj}=x_{j}$, $a\left(  x\right)  :=\prod_{1\leq i<j\leq N}\left(  x_{i}%
-x_{j}\right)  $, $dm$ is Haar measure and $c_{N}$ is a constant. The relation
to unistochastic matrices follows from the equation%
\[
\mathrm{Tr} \bigl(  D(x) U D(y) U^{\ast}\bigr)
=\sum_{i,j=1}^{N}x_{i}\left\vert U_{ij}\right\vert ^{2}y_{j}.
\]
The aforementioned linear operator is known algebraically, that is, with some
computational effort one can determine the action on any (low-degree!)
polynomial. For $N=3$ we were able to find a parametrized family of measures
to implement the operator, roughly%
\[
p\left(  x\right)  \mapsto\int_{U\left(  3\right)  }p\left(  xf\left(
U\right)  \right)  d\mu_{k}\left(  U\right)  ,
\]
where $f\left(  U\right)  _{ij}:=\left\vert U_{ij}\right\vert ^{2}$ as in
(\ref{unist}) for any polynomial $p$. It is known that
\[
\int_{U\left(  N\right)  }g\left(  \left\vert U_{ij}\right\vert ^{2}\right)
dm\left(  U\right)  =\left(  N-1\right)  \int_{0}^{1}g\left(  t\right)
\left(  1-t\right)  ^{N-2}dt,
\]
for any continuous function $g$ and any matrix entry $U_{ij}$. As in Section 5
we can compute the average (generalized) entropy for the entries $\left\vert
U_{ij}\right\vert ^{2}$ with respect to Haar measure ($q\neq1$):%
\begin{align*}
\left\langle S_{q}\right\rangle _{Haar}  &  =\frac{1}{N\left(  q-1\right)
}\int_{U\left(  N\right)  }\sum_{i,j=1}^{N}\left(  \left\vert U_{ij}%
\right\vert ^{2}-\left\vert U_{ij}\right\vert ^{2q}\right)  dm\left(  U\right)
\\
&  =\frac{N\left(  N-1\right)  }{q-1}\int_{0}^{1}\left(  t-t^{q}\right)
\left(  1-t\right)  ^{N-2}dt\\
&  =\frac{1}{q-1}\left(  1-\frac{N!}{\left(  q+1\right)  _{N-1}}\right) \\
&  =N!\sum_{i=0}^{N-2}\frac{\left(  -1\right)  ^{i}}{i!\left(  N-2-i\right)
!\left(  i+2\right)  \left(  q+i+1\right)  };
\end{align*}
the last equation is the partial fraction decomposition. Also $\lim
_{q\rightarrow1}\left\langle S_{q}\right\rangle =\psi\left(  N+1\right)
-\psi\left(  2\right)  =\sum_{j=2}^{N}\frac{1}{j}$, in agreement with the
known results for the mean entropy of random complex vectors distributed
according to the unitarily invariant measure \cite{Jo90,BZ06}.

We are thus tempted to speculate that there exists a measure $\mu_{k}$ on
$U\left(  N\right)  $ such that
\[
\int_{U\left(  N\right)  }g\left(  \left\vert U_{ij}\right\vert ^{2}\right)
d\mu_{k}\left(  U\right)  =\frac{\Gamma\left(  Nk\right)  }{\Gamma\left(
k\right)  \Gamma\left(  \left(  N-1\right)  k\right)  }\int_{0}^{1}g\left(
t\right)  t^{k-1}\left(  1-t\right)  ^{\left(  N-1\right)  k-1}dt,
\]
with $g,\left\vert U_{ij}\right\vert ^{2}$ as above. However we must emphasize
that there is an important difference between $\mathcal{U}_{3}$ and
$\mathcal{U}_{N}$ with $N\geq4$. There is a single inequality characterizing
$\mathcal{U}_{3}$ inside $\mathcal{B}_{3}$ (the condition is $Q\left(
b\right)  \geq0$) while $\left(  N-2\right)  ^{2}$ inequalities occur in
general \cite{Di06}. Another difference is that the elements of $U$ can not
necessarily be determined from the values $\left\{  \left\vert U_{ij}%
\right\vert ^{2}:1\leq i,j\leq N\right\}  $ (that is, up to left and right
multiplication by diagonal unitary matrices and permutation of rows or
columns). For instance for the flat matrix $W_{4}$ of van der Waerden, with
all entries equal to $1/4$ there exists a one parameter family of unitary
matrices $U(\alpha)$ (rescaled complex Hadamard matrices \cite{TZ06}) such
that $|[U(\alpha)]_{ij}|^{2}=[W_{4}]_{ij}=1/4$.

It would be interesting to be able to fit the \textquotedblleft
flat\textquotedblright\ measure on $\mathcal{U}_{N}$ (inherited from
$\mathcal{B}_{N}$) into the $\mu_{k}$ framework suggested above, but this
appears to be a sizable research problem in itself.


\end{document}